\documentclass[runningheads, envcountsame, a4paper]{llncs}

\usepackage{hyperref}
\usepackage{aliascnt}

\usepackage{thmtools}
\usepackage{thm-restate}
\usepackage{multicol}
\usepackage{amssymb}
\usepackage{graphicx}
\usepackage{bussproofs}
\usepackage{amsmath}

\usepackage[nameinlink]{cleveref}
\crefalias{definition}{Definition}

\usepackage{amsfonts}
\usepackage{xspace}
\usepackage{tikz-cd}
\usepackage{tikz}
\usepackage{pgf}
\usetikzlibrary{arrows,trees,positioning,arrows,chains,shapes.geometric, decorations.pathreplacing,decorations.pathmorphing,shapes,matrix,shapes.symbols,shapes,snakes,automata,backgrounds,petri,arrows.meta,positioning,arrows,calc,matrix,shapes.multipart,decorations.text,calc,arrows.meta}
\tikzstyle{nod}= [circle, draw,inner sep=0pt, minimum size=0.5cm]
\tikzstyle{m}=[circle, thin, draw, minimum size=14mm,inner sep=3pt]
\tikzset{reflexive left/.style={->,loop,looseness=10,in=160,out=220},
reflexive right/.style={->,loop,looseness=10,in=20,out=320}}

\newcommand{\agents}{\mathcal{A}}
\renewcommand{\implies}{\rightarrow}

\title{A Dynamic Deontic Simplicial Logic \\ for Joint Commitments}

\author{
  Giorgio Cignarale\orcidID{0000-0002-6779-4023}\inst{1} \and
  Hugo Rincon Galeana\orcidID{0000-0002-8152-1275}
}

\institute{
  TU Wien, Theory and Logic Group, Vienna, Austria. \\
  \email{giorgio@logic.at}
  \and
  \email{hugorincongaleana@gmail.com}
}

\titlerunning{Dynamic Deontic Simplicial Logic}
\authorrunning{G. Cignarale, H. Rincon Galeana}

\begin{document}
\maketitle

\begin{abstract}
We introduce the Deontic Simplicial Logic (\textbf{DSL}), a deontic logic for group obligations grounded in simplicial complexes: vertices encode individual commitments, and higher-dimensional simplices encode the joint commitments of the groups they connect. The resulting group modality behaves like a distributed-commitment operator with a genuinely normative character: it validates achievement but not the unrestricted introspection or monotonicity familiar from its epistemic counterpart, and impurity lets the model distinguish an agent's mere absence from a configuration from an explicit commitment to the contrary. We give a sound and complete axiomatization for the group modality.
%at arbitrary group size, using a stratified language and a Group Invariance lemma to overcome the standard obstacle to completeness for group operators of this kind. 
We then extend \textbf{DSL} to the Dynamic Deontic Simplicial Logic (\textbf{DDSL}), which introduces action modalities modeling agents' choices among mutually exclusive commitments, with effects captured by a product update construction on simplicial models; to our knowledge, this is the first dynamic deontic logic built on simplicial complexes. Soundness and completeness for \textbf{DDSL} are established via reduction axioms to the static case. Throughout, we illustrate both logics with worked examples of static and dynamic multi-agent commitment scenarios.
\end{abstract}

\textbf{Keywords:} Deontic logic, Simplicial complexes, Joint commitments

\section{Introduction}
\label{sec:Intro}

Reasoning about obligations and permissions in multi-agent settings is one of the central challenges in deontic logic \cite{Horty2001,Kooi2008-KOOMCB,Tamminga2013-TAMDLF,collobl,Tamminga2020-TAMTIO-6}. Traditional deontic frameworks treat agency and normativity largely at the propositional or relational level, offering limited tools to capture how collective commitments emerge from structured interactions among agents.
At the same time, distributed systems research has developed powerful geometric methods, most notably simplicial complexes \cite{comb_topo,topoAsyn}, to model partial information, coordination, and higher-order interaction among agents.
This paper proposes a new perspective, namely interpreting simplicial complexes as models of normative structure. 
Instead of viewing simplices as encoding shared knowledge, we treat them as representing shared commitments among agents.
The central advantage of our geometric approach is expressiveness that is unavailable in Kripke-based frameworks.\footnote{The advantages of the impure approach in the epistemic settings are shown in \cite{rojo}.} In standard multi-agent deontic logics, the absence of a commitment between agents \textit{must} be encoded as a formula. In the simplicial approach, the absence of a simplex connecting agents is a structural property of the model representing the impossibility of those agents to be jointly committed, which can also be expressed with formulas.

\textbf{Survey of the literature:}
Combinatorial topology has been widely used in distributed computing, especially to model concurrency and asynchrony \cite{comb_topo,topoAsyn} and the epistemic interpretation of simplicial complexes is well established~\cite{castanedaDagRep,DEL-cover,know_in_simpl,goubault_chroma,Castaneda2023-CASCPL}.
Formally, a simplicial complex consists of a set of vertices together with a collection of simplices, where a simplex is a finite set of vertices. In particular, every subset of a simplex is itself a simplex and every vertex appears as a simplex on its own. A simplex is called \textit{facet}, if it is maximal (i.e. there is no proper superset). The dimension of a simplex is the number of vertices minus one.
Some examples of simplices are offered in \autoref{fig:SCdim}.

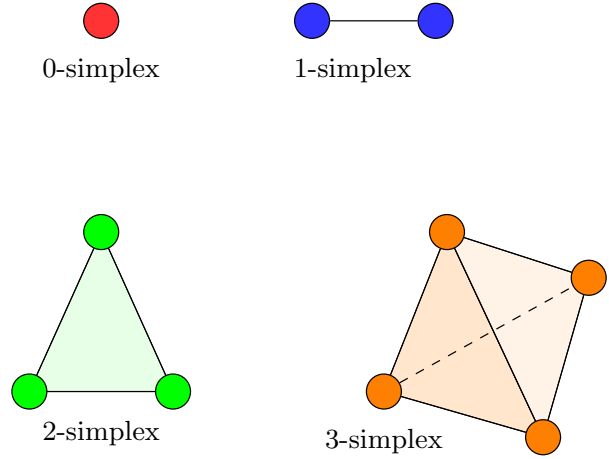
\begin{figure}[t]
    \centering
    \resizebox{8cm}{!}{%
\begin{tikzpicture}
\begin{scope}[xshift=-6cm,>=stealth, yshift=3cm]

    \node[circle,draw,minimum size=4mm, fill=red!80] (1) {};
    \node (a) [below = .1cm  of 1] {$0$-simplex};

    \node[circle,draw,minimum size=4mm, fill=blue!80] (2) [right = 2cm  of 1] {};
    \node[circle,draw,minimum size=4mm, fill=blue!80] (3) [right = 1cm  of 2] {};
    \node (b) [below = .1cm  of 2] {};
    \node (b1) [right = 1.3cm  of a] {$1$-simplex};

     \path[]
        (2) edge node[above] {} (3);

    \node[circle,draw,minimum size=4mm, fill=green] (4) [below = 2cm  of 1] {}; 
    \node (g) [below = 1.5cm  of 4] {}; 
    \node[circle,draw,minimum size=4mm, fill=green] (5) [left = .5cm  of g] {};
    \node[circle,draw,minimum size=4mm, fill=green] (6) [right = .5cm  of g] {};
    \node (c) [below = .1cm  of g] {$2$-simplex};

    \path[]
        (4) edge node[above] {} (5)
        (5) edge node[above] {} (6)
        (4) edge node[above] {} (6);

     \begin{pgfonlayer}{background}
\draw[fill=green!10,line width=0.1mm,line cap=round,line join=round] (4.center)--(5.center)--(6.center)--cycle;
\end{pgfonlayer}

    \node[circle,draw,minimum size=4mm, fill=orange] (7) [right = 2cm  of 6] {}; 
    \node (e) [right = 1.5cm  of 7] {}; 
    \node[circle,draw,minimum size=4mm, fill=orange] (8) [below = .2cm  of e] {};
    \node (f) [above = 1.5cm  of 8] {}; 
    \node[circle,draw,minimum size=4mm, fill=orange] (9) [right = .2cm  of f] {};
    \node (w) [above = 1.5cm  of 7] {};
    \node[circle,draw,minimum size=4mm, fill=orange] (10) [right = .4cm  of w] {};
    \node (d) [below = .1cm  of 7] {$3$-simplex};

    \path[]
        (7) edge node[above] {} (8)
        (8) edge node[above] {} (9)
        (9) edge node[above] {} (10)
        (7) edge node[above] {} (10)
        (8) edge node[above] {} (10)
        ;

    \path[dashed]
        (7) edge node[above] {} (9);
    
    \begin{pgfonlayer}{background}
\draw[fill=orange!20,line width=0.1mm,line cap=round,line join=round] (7.center)--(8.center)--(10.center)--cycle;
\draw[fill=orange!10,line width=0.1mm,line cap=round,line join=round] (9.center)--(8.center)--(10.center)--cycle;
\end{pgfonlayer}
    
    \end{scope}

\end{tikzpicture}

    }
    \caption[Simplices of various dimensions.]{Simplices of various dimensions.}
    \label{fig:SCdim}
\end{figure}
In the epistemic framework, each vertex represents the local state of an agent, while simplices represent global configurations made of the local states of agents. Agents' uncertainty is represented by having multiple simplices that share that agent vertex, and higher-dimensional topological properties are used to represent group knowledge.
In particular, \textit{impure} simplicial complexes, i.e., simplicial complexes whose maximal simplices do not all have the same dimension, have been employed in modeling crashing agents \cite{Rojo_bisim,rojo,Deadoralive,Rojo-2-3,goubault_et_al:LIPIcs.STACS.2022.33}. In our deontic setting, by contrast, impurity models the deliberate absence of commitment: an agent whose vertex is missing from a simplex has not committed to the others in that configuration, which is normatively distinct from an agent who has committed negatively.

Despite their success in epistemic and computational contexts, these topological tools have so far received no attention as a foundation for deontic reasoning.
Work on agency in deontic logic has instead largely focused on the `Seeing To It That' (STIT) frameworks, which analyze obligations relative to the actions available to agents at a given moment \cite{Horty2001,sep-logic-action}. Extensions of STIT to collective settings have also been proposed to study joint obligations and shared responsibility \cite{Horty1995,Tamminga2021-TAMERF}. 
STIT logic in Kripke frameworks has been proved to be undecidable for group of agents \cite{herzig:hal-03526736}.

The notion of joint commitment has been studied extensively in the philosophy of collective intentionality.
\cite{Castelfranchi} argues that social commitments between agents are what hold a group together, and that no organization is possible without them. On this account, commitment is not merely a philosophical construct but a structural feature of multi-agent interaction: social commitments are between agents, typically with respect to some action or state which the committed agent is prepared to perform or achieve. Crucially, \cite{Castelfranchi} distinguishes individual from collective commitments, where the latter express the shared purpose or objective of the group as a whole, a distinction our simplicial framework is designed to capture geometrically, with vertices representing individual commitments and higher-dimensional simplices encoding their joint counterparts.
Notions of joint commitments and their dynamics have also received a logical characterization \cite{Herzig-commit,Lorini-commit,MEYER19991}.
In particular, \cite{Lorini-commit} proposes a model-theoretic semantics and a complete logic for the dynamics of commitments, formalizing social commitment within STIT logic and extending it with dynamic operators for commitment creation and cancellation. 
Our work is related to Lorini's \cite{Lorini-commit} in its use of STIT-style agency and its dynamic extension, but departs from it by grounding the semantics in simplicial complexes rather than Kripke structures, which allows the geometric structure of the model itself to encode the presence or absence of joint commitment among subgroups of agents.

\textbf{Main contributions:} 
We adopt a normative interpretation of simplicial complexes, in which simplices represent collective commitments among agents and we formulate the Deontic Simplicial Logic (\textbf{DSL}) to express these joint commitments in a geometric interpretation: vertices represent individual commitments of agents, while higher-dimensional simplices encode joint commitments shared by the agents corresponding to their vertices.

We introduce the group modality $\mathbf{D}_G \varphi$ representing a joint commitment of agents in $G$ towards $\varphi$.
Using impure simplicial complexes makes the absence of a simplex normatively meaningful, reflecting the lack of mutual commitment in a given configuration.
In the spirit of STIT-based agency \cite{Horty2001}, we interpret joint commitment as an achievement notion: a group $G$ is jointly committed to $p$ in the sense that the members of $G$ have collectively arranged their choices so that they see it to that $p$ holds.

The static \textbf{DSL} captures the normative landscape at a single moment: which commitments hold, which are absent, and how they are distributed across groups. However, it cannot represent how that landscape changes when a commitment is made, withdrawn, or changed.
To model this, we extend \textbf{DSL} to the Dynamic Deontic Simplicial Logic (\textbf{DDSL}), which introduces action modalities $[\alpha_G]\varphi$ expressing that $\varphi$ holds after the agents in group $G$ jointly establish their commitments as arranged in $\alpha_G$. The underlying update mechanism adapts the product update of Dynamic Epistemic Logic (DEL) \cite{plaza,Baltag1998,ditmarsch2007dynamic,Gob_DEL} to the simplicial setting, filtering the simplicial complex to retain only those faces compatible with agents' new commitments. This gives DDSL the ability to reason about commitment creation, cancellation, and the restructuring of group normative relationships over time.

We prove soundness and completeness of both logics: the static case is proved via the canonical simplicial model construction \cite{rojo} and the dynamic case via reduction axioms to the former.
We also provide formalizations of examples both in the static and dynamic cases.

\textbf{Motivating example:}
Alex ($a$) wants to invite Baldo ($b$) and Carla ($c$) for a small party $p$. To communicate, they mainly use a group chat. We present five different scenarios:

    \begin{itemize}
    
        \item  \textit{Version 1:} $b$ says that he will participate, but $c$ says she cannot. This means that $a$ and $b$ have a joint commitment to come, while $c$ has a joint commitment to not come.\footnote{We argue that there is a commitment of agent $c$ in \textit{not} coming to the party if $c$ said so. If $c$ decides to come to the party after all, $c$ cannot expect there to be enough room, for example.}

        \item  \textit{Version 2:} $b$ says that he will participate, but $c$ does not reply. This means that $a$ and $b$ have a joint commitment to come, while $c$ has no joint commitment to either option.

        \item \textit{Version 3:} Instead of using the group chat, each pair communicates separately, and agents reveal that they will join the party. As a result, each agent has only a pairwise commitment to each other but no joint commitment to the whole group.

        \item  \textit{Version 4:} $b$ says that he will participate, but only if $c$ also participate, and $c$ states she will participate only if $b$ does. In case $b$ (resp. $c$) will not participate, $c$ (resp. $b$) will also not participate. $a$, being the host, will participate anyway. This means that either $b$ and $c$ are both committed to come, or that $b$ and $c$ are both not committed to come, but in no case they have different commitments.

        \item \textit{Version 5:} $b$ says that he will participate, but only if $c$ also participate, and the same holds for $c$ w.r.t. $b$.
        Then $b$ privately communicates to $a$ that $b$ will not participate if $c$ does not let them know. The result is similar to the above, but $b$ also has a private commitment towards $a$ of not coming in case $c$ does not let them know.

        \item \textit{Version 6:} $c$ replies positively to the invitation, while $b$ tells them that he will let them know. In this version, $a$ and $c$ are committed to come, no matter what $b$ does. We interpret the `let you know' as $b$ being committed to both options, one in each possible configuration, and is only a matter of time before $b$ decides to come or not, realizing a single configuration.

    \end{itemize}

Further dynamic multi-agent examples are provided in later sections.

\textbf{Overview:} \Cref{sec:DSL} introduces the static Deontic Simplicial Logic \textbf{DSL} and the formalization of the motivating example using that framework. Its soundness and completeness are proved in \Cref{sec:SC1}. Some properties of deontic simplicial models are provided in \Cref{sec:prop}. \Cref{sec:DDSL} extends the static \textbf{DSL} logic to the dynamic \textbf{DDSL}, 
and provides further examples involving (joint) actions.
Its soundness and completeness are proved in \Cref{sec:SC2}. Some conclusions are provided in \Cref{sec:concl}.

\section{Deontic Logic for Simplicial Complexes}
\label{sec:DSL}

We consider a finite set of agents, $\mathcal{A}=\{\ a,b,\ldots\}$ such that $\vert \mathcal{A} \vert = n+1$ for $n \in \mathbb{N}$, and a set $ P = \bigsqcup_{a \in \mathcal{A}} P_a$ of propositional variables where each $P_a$ is a countable set of propositional variables indexed by an agent $a \in \mathcal{A}$\footnote{In line with the standard combinatorial topology, the number $|\mathcal{A}|$ of agents is $ n + 1$ in order to make the dimension of the simplicial complexes' equal to $n$.}. Furthermore, we assume that the agent-indexed propositional sets are pairwise disjoint, i.e. if $a \neq b \in \mathcal{A}$, then $P_a \cap P_{b} = \varnothing$.

Given the deontic nature of our framework, our propositions are used to indicate agent's adherence to a certain plan or action for which there might be some (joint) commitment.
For example, if  $p$ stands for `meeting at 15:00', then $p_a$ stands for `meeting at 15:00 for agent $a$', $p_b$ is `meeting at 15:00 for agent $b$' and so on.

We begin with the definition of a simplicial model:

\begin{definition}[Simplicial model \cite{Deadoralive}]
\label{def:SimMod}
A simplicial model $\mathcal{C}$ is a triple $(C, \chi, l)$ consisting of:

\begin{itemize}
    \item A (simplicial) complex $C \ne \varnothing$ is a collection of simplices that are non-empty finite subsets of a given set $V$ of vertices such that $C$ is downward closed (i.e. $X \in C$ and $\varnothing \ne Y \subseteq X$ imply $Y \in C$). Simplices represent state of affairs in which agents hold mutual commitments to each other.
    \item Vertices represent local states of agents, with a chromatic map $\chi: V \rightarrow \mathcal{A}$ assigning each vertex to one of the agents in such a way that each agent has at most one vertex per simplex, i.e. $\chi(v) = \chi(u)$ for some $v, u \in X \in C$ implies that $v = u$. For $X \in C$, we define $\chi(X) = \{\chi(v) | v \in X\}$ to be the set of agents in simplex $X$.
    \item A valuation $l : V \rightarrow 2^P$ assigns to each vertex which propositional variables of the vertex’s owner are true in it, i.e. $l(v) \subseteq P_a$ whenever $\chi(v) = a$. Variables from $P_a \setminus l(v)$ are false in vertex $v$, whereas variables from $P \setminus P_a$ do not belong to agent $a$ and cannot be evaluated in $a$’s vertex $v$. The set of variables that are true in a simplex $X \in C$ is given by $l(X):= \bigsqcup_{v \in X}l(v)$
    
\end{itemize}
\end{definition}

If $Y \subseteq X$ for $X, Y \in C$, we say that $Y$ is a \textit{face} of $X$. Since each simplex is a face of
itself, we use ‘simplex’ and ‘face’ interchangeably. A face $X$ is a \textit{facet} if it is a maximal
simplex in $C$, i.e., $Y \in C$ and $Y \supseteq X$ imply $Y = X$.
The set of facets is denoted $\mathcal{F}(C)$. The dimension of simplex $X$, written $dim(X) = |X| - 1$, e.g., vertices are of dimension $0$, edges are of dimension $1$, etc. The dimension of a simplicial complex (model) is the largest dimension of its facets. A simplicial complex is \textit{pure} if all facets have dimension $n$, i.e. contain vertices for all agents. Otherwise it is \textit{impure}.
A pointed simplicial model is a pair $(\mathcal{C}, X)$ where $X \in C$.
In the following, we focus on (but not limited to) impure simplicial complexes, and we use the term configuration to refer to the simplices of maximal dimension representing possible (joint) commitments of agents.
Since we adopt a deontic interpretation of simplicial complexes, we sometimes refer to them as \textit{deontic simplicial complexes}.

We can now introduce the elements of our deontic logic on simplicial complexes.
In contrast with the standard epistemic interpretation of simplicial models, which focuses on knowledge of single agents, we use a group modality as our primary modality, which is meant to capture the mutual commitment of agents in a group $G \subseteq \agents$. We start by defining the language for the deontic simplicial logic \textbf{DSL}:

\begin{definition}[Stratified language $\mathcal{L}^{DSL}$]
\label{def:langL}
For $G \subseteq \agents$ (possibly empty), the fragment $\mathcal{L}^{DSL}_G$ is defined by mutual recursion on $G$:
\[
\varphi ::= \top \ \mid\ p_a\ (a\in G) \ \mid\ \neg\varphi \ \mid\ (\varphi\wedge\varphi) \ \mid\ \mathbf{D}_H\varphi\ \ (\varnothing\neq H\subseteq G,\ \varphi\in\mathcal{L}^{DSL}_H).
\]
We set $\bot := \neg\top$, and define $\rightarrow,\vee$ as usual. The \emph{full language} is $\mathcal{L}^{DSL} := \mathcal{L}^{DSL}_\agents$. A formula $\varphi \in \mathcal{L}^{DSL}_G$ is called \emph{$G$-stratified}. 

\end{definition}

The novel modality $\mathbf{D}_G \varphi$ reads as `agents in group $G \subseteq \mathcal{A}$ have mutual (or joint) commitment towards $\varphi$'. 
For example, interpreting $\varphi$ as `meeting at 15:00', $\mathbf{D}_{\{a,b\}} \varphi$ is read as `$a$ and $b$ have mutual commitment to meet at 15:00'. Following \cite{Gilbert1989}, we intend this commitment to be explicit and known to the agents, as resulting from some sort of (previous) communication, making our notion of commitment deeply linked to a pre-existing epistemic background, upon which commitments are built. This is not to be seen as a bug, but rather a feature: group normativity is based on a common understanding of the norms, rules and other normative constraints \cite{Tuomela2007,Bicchieri_2005}.

This means that when such epistemic background is missing, joint commitments are affected as well. Consider the case where agent $a$ and $c$ have a joint commitment towards $p$ because of a previous agreement. Agent $b$ might as well be committed to some proposition, but since there was no joint communication, $b$ has no \textit{joint} commitment towards it.\footnote{An example of this situation is offered in \autoref{fig:disa} (right).}

The $G-$stratification restricts the scope of formulas that can be present within the modality. Its informal reading is that a group of agents $H$ cannot be committed to formulas involving agents of other groups $G \supset H$. In other words, a group commitment is made by the commitments of its participants, and not of those of other agents.

\begin{lemma}[Monotonicity of stratification]
\label{lem:strat-mono}
If $H \subseteq G$, then $\mathcal{L}^{DSL}_H \subseteq \mathcal{L}^{DSL}_G$.
\end{lemma}

\begin{proof}
By induction on $\varphi \in \mathcal{L}^{DSL}_H$. If $\varphi=\top$, then $\varphi\in\mathcal{L}^{DSL}_G$ trivially, by the same clause. If $\varphi=p_a$ with $a\in H$, then $a\in H\subseteq G$, so $\varphi\in\mathcal{L}^{DSL}_G$. The cases $\neg\psi$, $(\psi_1\wedge\psi_2)$ follow immediately from the induction hypothesis. If $\varphi=\mathbf{D}_K\psi$ with $\varnothing\neq K\subseteq H$ and $\psi\in\mathcal{L}^{DSL}_K$, then $K\subseteq H\subseteq G$, so $\varphi\in\mathcal{L}^{DSL}_G$ by the same clause, taking the same $K$.
\end{proof}

\begin{definition}[Semantics]
\label{def:semL}
For a group $G \subseteq \mathcal{A}$:
   \begin{itemize}
   \item $\mathcal{C},X \models \top \quad \text{always}$;
    \item $\mathcal{C},X \models p_a \text{ iff } 
    %a \in \chi(X) \text{ and } 
    a \in \chi(X) \text{ and } p_a \in l(X)$, 
    \item $\mathcal{C},X \models \neg \varphi \text{ iff } \mathcal{C},X \not\models \varphi $
    \item $\mathcal{C},X \models (\varphi \wedge \psi) \text{ iff } \mathcal{C},X \models \varphi \text{ and } \mathcal{C},X \models \psi $
    
    \item $\mathcal{C},X \models \mathbf{D}_G \varphi \text{ iff } G \subseteq \chi (X) \text{ and }  \mathcal{C}, Y \models \varphi \text{ for all } Y \in C \text{ s.t. }  
    G \subseteq \chi (X \cap Y)$

\end{itemize}

We say that a formula $\varphi$ is valid, written $\mathcal{C} \models \varphi$, iff for all $X \in C$, $\mathcal{C}, X \models \varphi$.
\end{definition} 

Among the well-known group modalities in epistemic logic, such as mutual knowledge, distributed knowledge and common knowledge \cite{van2015handbook}, 
the semantics of our joint commitment operator is very similar to the semantics of distributed knowledge
\cite{know_in_simpl,Deadoralive,ericDK}. 
In the epistemic framework, distributed knowledge is the knowledge obtained by joining the local states of all agents.
This `centralized' feature, treating the group as a `higher-order agent', makes it the best candidate for interpreting joint commitment, as we intend it as a (shared) obligation towards each agent in the group and also to the group itself.
The crucial difference from the standard account is that we require $G \subseteq \chi(X)$, as agents' presence is not a given in impure simplicial complexes.

From a semantic perspective, the mutual commitment $\mathbf{D}_G \varphi$ is true if $\varphi$ is true in all simplices that share a $G$-colored face with it. This means that when evaluating a formula for a pair, such as $\mathcal{C},X \models \mathbf{D}_{\{a,b\}}\varphi$, we will check all the adjacent faces where $\varphi$ holds along the edges of $X$ that are \{a,b\}-colored.

As an example, consider the complex $\mathcal{C}$ in \autoref{fig:EX6} (left). We want to check whether the formula $\mathcal{C},X \models \mathbf{D}_{\{a,b,c\}}(p_a \wedge p_b \wedge p_c)$ is true. Firstly we check that $\{a,b,c\} \subseteq \chi(X)$, which holds. Secondly, we need to look for all the adjacent faces $Y$ such that
$\{a,b,c\} \subseteq \chi(Y\cap X)$ and $C,Y \models p_a \wedge p_b \wedge p_c$.
In this example, the only face that satisfies both conditions is $X$ itself, making the formula true.

A single agent version can be obtained by using singleton groups: $\mathcal{C},X \models \mathbf{D}_{\{a\}} \varphi \text{ iff } a \in \chi(X) \text{ and } \mathcal{C}, Y \models \varphi \text{ for all } Y \text{ s.t. } a \in \chi (X \cap Y)$, which semantically corresponds to the definition of knowledge in the epistemic framework.
In our interpretation, $\mathbf{D}_{\{a\}} \varphi$ represents the commitment that an agent has to herself.

Because of the distributed nature of our operator, agents might be jointly committed to state of affairs, without having a direct single-agent commitment towards that state of affairs. An example is shown in \autoref{fig:bigtri}. $C,X \not \models D_{\{a\}} (p_a \wedge p_b)$, as $C,Y \models \neg p_b$, and $C,X \not \models D_{\{b\}} (p_a \wedge p_b)$, as $C,W \models \neg p_a$. However, $C,X \models D_{\{a,b\}} (p_a \wedge p_b)$, as $(p_a \wedge p_b)$ holds both in $X$ and $Z$.

\begin{figure}[t]
    \centering
    \resizebox{7cm}{!}{%
\begin{tikzpicture}
    \begin{scope}[xshift=-6cm,>=stealth, yshift=3cm]
    \node[circle,draw,minimum size=10mm, fill=red] (1) {$p_a$};
    \node[circle,draw,minimum size=10mm, fill=blue!50] (2) [right = 2cm  of 1]  {$p_b$};
    \node (g) [right = 1cm  of 1]  {};
    \node[circle,draw,minimum size=10mm, fill=green] (3) [above = 2cm  of g] {$\neg p_c$}; 
    \node (h) [above = 0.5cm  of g]  {$X$};
    \node[circle,draw,minimum size=10mm, fill=green] (4) [below = 2cm  of g] {$p_c$};
    \node (r) [right = 2cm  of 2]  {};
    \node[circle,draw,minimum size=10mm, fill=red] (5) [below= 1cm  of r]  {$\neg p_a$};
    \node[circle,draw,minimum size=10mm, fill=green] (6) [above = 1cm  of r]  {$\neg p_c$};
     \node (l) [left = 2cm  of 1]  {};
    \node[circle,draw,minimum size=10mm, fill=green] (7) [below= 1cm  of l]  {$\neg p_c$};
    \node[circle,draw,minimum size=10mm, fill=blue!50] (8) [above = 1cm  of l]  {$\neg p_b$};
    \node (h'') [right = 1cm  of 2]  {$W$};
    \node (h') [below = 0.5cm  of g]  {$Z$};
    \node (h''') [left = 1cm  of 1]  {$Y$};

    \node (s) [below = 0.3cm  of h']  {};
    %\node (h'') [left = 1cm  of s]  {$W$};
    %\node (h''') [right = 1cm  of s]  {$Z$};
    \node (f) [above = 1cm  of 1]  {\Large{$\mathcal{C}$}};

    \path[]
        (1) edge node[above] {} (2)
        (1) edge node[above] {} (3)
        (3) edge node[above] {} (2)
        (4) edge node[above] {} (2)
        (1) edge node[above] {} (4)
        %(4) edge node[above] {} (5)
        %(4) edge node[above] {} (6)
        (5) edge node[above] {} (6)
        (2) edge node[above] {} (5)
        (2) edge node[above] {} (6)
        (7) edge node[above] {} (8)
        (1) edge node[above] {} (7)
        (1) edge node[above] {} (8)
        ;

\begin{pgfonlayer}{background}
\draw[fill=black!10,line width=0.1mm,line cap=round,line join=round] (1.center)--(2.center)--(3.center)--cycle;
\draw[fill=black!10,line width=0.1mm,line cap=round,line join=round] (1.center)--(2.center)--(4.center)--cycle;
\draw[fill=black!10,line width=0.1mm,line cap=round,line join=round] (5.center)--(2.center)--(6.center)--cycle;
\draw[fill=black!10,line width=0.1mm,line cap=round,line join=round] (7.center)--(1.center)--(8.center)--cycle;
\end{pgfonlayer}

    \end{scope}
\end{tikzpicture}
    }
    \caption{Simplicial model $\mathcal{C}$ with four possible configurations $Y,X,Z,W$.}
    \label{fig:bigtri}
\end{figure}
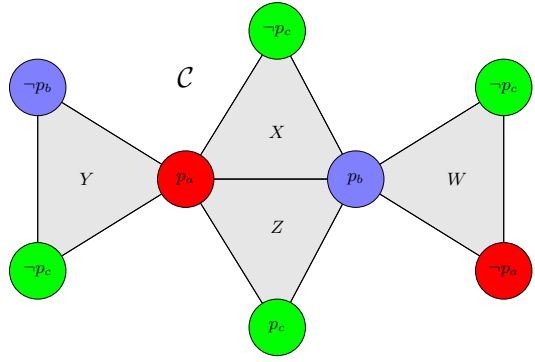

Importantly, the fact that $\mathbf{D}_G \varphi$ is true only if the agents in $G$ are actually present in the facet of evaluation makes sure that agents which are not present in the facet are not jointly committed to anything, distinguishing between an agent having a negative joint commitment (such as saying `I will not come to your party') and no joint commitment at all (such as not answering the phone at all).

Before proceeding, we establish two structural properties of simplicial models that will underpin the completeness proofs of \Cref{sec:SC1,sec:SC2}:

\begin{lemma}[Chromatic transitivity]
\label{lem:chromtrans}
For any simplicial model $\mathcal{C}=(C,\chi,l)$, any $\varnothing\ne G\subseteq\agents$, and any $X,Y,Z\in C$: if $G\subseteq\chi(X\cap Y)$ and $G\subseteq\chi(Y\cap Z)$, then $G\subseteq\chi(X\cap Z)$.
\end{lemma}
\begin{proof}
Fix $a\in G$. From $a\in\chi(X\cap Y)$: $Y$'s (unique, by chromaticity) $a$-vertex $v_a$ satisfies $v_a\in X$. From $a\in\chi(Y\cap Z)$: $Y$'s $a$-vertex, the same $v_a$, by uniqueness, satisfies $v_a\in Z$. Hence $v_a\in X\cap Z$, so $a\in\chi(X\cap Z)$. As $a\in G$ was arbitrary, $G\subseteq\chi(X\cap Z)$.
\end{proof}

\begin{lemma}[Group Invariance]
\label{lem:groupinv}
Let $\mathcal C=(C,\chi,l)$ be \emph{any} simplicial model, $\varnothing\ne G\subseteq\agents$, and $\psi\in\mathcal L^{DSL}_G$. For any $X,X'\in C$ with $G\subseteq\chi(X\cap X')$:
\[
\mathcal C,X\models\psi \iff \mathcal C,X'\models\psi.
\]
\end{lemma}

\begin{proof}
By structural induction on $\psi$ (for whichever $G$ makes $\psi\in\mathcal L^{DSL}_G$ true; a single $\psi$ may lie in several such fragments, and the claim is proved for each).

\textbf{$\psi=\top$:} trivial.

\textbf{$\psi=p_a$}: since $G\subseteq\chi(X\cap X')$, in particular $a\in\chi(X\cap X')$, so $X,X'$ share the same (unique) $a$-vertex $v_a$. Since $P_a\cap P_b=\varnothing$ for $b\ne a$, only $v_a$ can contribute $P_a$-atoms to $l(X)$ or $l(X')$; hence $p_a\in l(X)\iff p_a\in l(v_a)\iff p_a\in l(X')$.

\textbf{$\psi=\neg\chi$, $\psi=\chi_1\wedge\chi_2$:} immediate from the IH.

\textbf{$\psi=\mathbf D_H\chi$}: Since $H\subseteq G\subseteq\chi(X\cap X')$, in particular $H\subseteq\chi(X\cap X')$, so both gates $H\subseteq\chi(X)$, $H\subseteq\chi(X')$ hold. It remains to show the two universally-quantified witness sets coincide:
\[
\{Z\in C: H\subseteq\chi(X\cap Z)\} = \{Z\in C: H\subseteq\chi(X'\cap Z)\}.
\]
Suppose $H\subseteq\chi(X\cap Z)$. Since $H\subseteq\chi(X\cap X')=\chi(X'\cap X)$, \autoref{lem:chromtrans} gives $H\subseteq\chi(X'\cap Z)$. The reverse inclusion is symmetric. Hence
$\mathcal C,X\models\mathbf D_H\chi \iff  \mathcal C,X'\models\mathbf D_H\chi$.

\end{proof}

Group Invariance says that any $G$-stratified formula, which is a formula built only from $G$'s own atoms and $\mathbf D_H$-modalities for $H\subseteq G$, cannot distinguish between two configurations that agree on all of $G$'s vertices, however much else about those configurations differs.

\subsection{\textit{Formalized Motivational Examples}}

With the tools introduced above, we can model our motivational example using simplicial complexes:

\begin{example}[Motivational example (formalized)]
\label{ex:meeting_formal} 

We formalize each version separately:

\begin{itemize}

    \item Version 1 is represented in \autoref{fig:EX1} (left). $a$ and $b$ state they will participate while $c$ will not, making $\mathcal{C},X \models \mathbf{D}_{\{a,b,c\}} (p_a \wedge p_b \wedge \neg p_c)$ true.

    \item Version 2 is represented in \autoref{fig:EX1} (middle). $a$ and $b$ will participate, $c$ did not answer. Thus $\mathcal{C}',X' \models \mathbf{D}_{\{a,b\}} (p_a \wedge p_b)$ holds, but any commitment also involving $c$ will result to be false: $\mathcal{C}',X' \not\models \mathbf{D}_{\{a,b,c\}} (p_a \wedge p_b \wedge p_c)$ and $\mathcal{C}',X' \not\models \mathbf{D}_{\{a,b,c\}} (p_a \wedge p_b \wedge \neg p_c)$, as $c \notin \chi(X')$. 
    In particular, absence of $c$ does not imply a commitment to $\neg p_c$: $\mathcal{C}',X' \models \neg p_c$ since $c \notin \chi(X')$, but $\mathcal{C}',X' \not\models \mathbf{D}_{\{c\}} \neg p_c$. This corresponds to the fact that $c$ will not attend the party, but has no commitment towards $a$ and $b$ about it.

    \item Version 3 is represented in in complex $\mathcal{C}''$ of \autoref{fig:EX1} (right)\footnote{
    In a more `real-world' reading, the agents in \autoref{fig:EX1} (right) may be seen as banks coordinating market-conduct rules, but reaching only pairwise agreements rather than a common strategy.
    }. We use $\overline{AB}$ to indicate the face of dimension one linking vertices labeled with the same lowercase letter. It is easy to check that $\mathcal{C}'', \overline{AB} \models \mathbf{D}_{\{a,b\}} (p_a \wedge p_b)$ (respectively for other pairs), while $\mathcal{C}'', \overline{AB} \not\models \mathbf{D}_{\{a,b,c\}} (p_a \wedge p_b \wedge p_c)$, as $c \notin \chi(\overline{AB})$.

    \item Version 4 is represented in complex $\mathcal{C}$ of \autoref{fig:EX6} (left). Each simplex represents a possible configuration: $Y$ represents the case where $c$ and $b$ do not come to the party in case one of them defects, and $X$ illustrates the scenario where they both come to the party. The two separate commitments are captured by the following formulas: $\mathcal{C},Y \models \mathbf{D}_{\{a,b,c\}} (p_a \wedge \neg p_b \wedge \neg p_c)$ and $\mathcal{C},X \models \mathbf{D}_{\{a,b,c\}} (p_a \wedge p_b \wedge p_c)$. 
    Since $a$ is the host, they will participate anyway: $\mathcal{C},X \models \mathbf{D}_{\{a\}} p_a $ as $a \in \chi(X \cap Y)$.

    \item Version 5 is represented in complex $\mathcal{C}'$ of \autoref{fig:EX6} (center). As above, each simplex represents a possible configuration of their meeting. 
    The difference with respect to the previous version is the addition of face $W'$, representing the result from $b$ telling $a$ that if $c$ is absent ($c \notin \chi(W')$) then $b$ will not come to the party $\neg p_b$.
    
    The fact that $c \notin \chi(W')$ makes $\mathcal{C}',W' \models \mathbf{D}_{\{a,b,c\}} (p_a \wedge p_b \wedge \neg p_c)$ false.
    Furthermore, also $\mathcal{C}', Y' \models \mathbf{D}_{\{a,b\}} (p_a \wedge p_b \wedge \neg p_c)$ is false, because while $\{a,b\} \subseteq \chi(W')$ and $\{a,b\} \subseteq \chi(Y' \cap W')$, $(p_a \wedge p_b \wedge \neg p_c)$ is false in $W'$.

    \item Version 6 is represented in $\mathcal{C}''$ of \autoref{fig:EX6} (right), with two possible configurations, $X''$ where $b$ decides to join and $Y''$ where he doesn't. In both $a$ and $c$ are committed to go to the party: $\mathcal{C},X\models \mathbf{D}_{\{a,c\}}(p_a \wedge p_c)$, as $\{a,c\} \subseteq \chi(X'' \cap Y'')$ and $p_a \wedge p_c$ holds in both $X''$ and $Y''$.
\end{itemize}

\end{example}

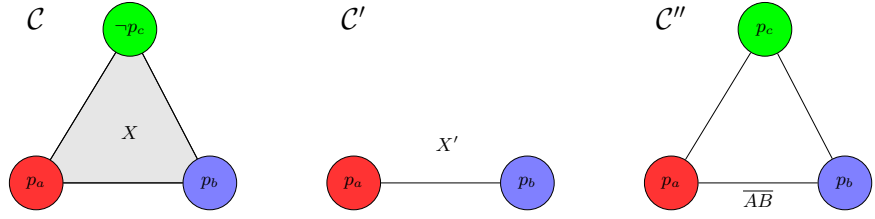
\begin{figure}[t]
    \centering
    \resizebox{11.5cm}{!}{%
\begin{tikzpicture}
    \begin{scope}[xshift=-6cm,>=stealth, yshift=3cm]
    \node[circle,draw,minimum size=9mm, fill=red!80] (1) {$p_a$};
    \node[circle,draw,minimum size=9mm, fill=blue!50] (2) [right = 2cm  of 1]  {$p_b$};
    \node (g) [right = 1cm  of 1]  {};
    \node[circle,draw,minimum size=9mm, fill=green] (3) [above = 2cm  of g] {$\neg p_c$}; 
    \node (h) [above = 0.5cm  of g]  {$X$};
    \node (h') [above = 2cm  of 1]  {\Large{$\mathcal{C}$}};

    \path[]
        (1) edge node[above] {} (2)
        (1) edge node[above] {} (3)
        (3) edge node[above] {} (2);

\begin{pgfonlayer}{background}
\draw[fill=black!10,line width=0.1mm,line cap=round,line join=round] (1.center)--(2.center)--(3.center)--cycle;
\end{pgfonlayer}

    \node[circle,draw,minimum size=9mm, fill=red!80] (4) [right = 1.5cm  of 2] {$p_a$};
    \node[circle,draw,minimum size=9mm, fill=blue!50] (5) [right = 2cm  of 4]  {$p_b$};
    \node (g) [right = 1cm  of 4]  {};
    \node (h) [above = 0.3cm  of g]  {$X'$};
    \node (h') [above = 2cm  of 4]  {\Large{$\mathcal{C}'$}};

    \path[]
        (4) edge node[above] {} (5);

    \node[circle,draw,minimum size=9mm, fill=red!80] (6) [right = 1.5cm  of 5] {$p_a$};
    \node[circle,draw,minimum size=9mm, fill=blue!50] (7) [right = 2cm  of 6]  {$p_b$};
    \node (g) [right = 1cm  of 6]  {};
    \node[circle,draw,minimum size=9mm, fill=green] (8) [above = 2cm  of g] {$p_c$}; 
    \node (h') [above = 2cm  of 6]  {\Large{$\mathcal{C}''$}};
    %\node (h) [above = 0.5cm  of g]  {$X$};

    \path[]
        (6) edge node[below] {$\overline{AB}$} (7)
        (6) edge node[above] {} (8)
        (8) edge node[above] {} (7);

    \end{scope}
\end{tikzpicture}
    }
    \caption[Simplicial models representing different examples from \autoref{ex:meeting_formal}, from left to right: versions 1, 2 and 3.]{Pure simplicial model $\mathcal{C}$ where every agent has a commitment to each other and to the group (left).
    Impure simplicial model $\mathcal{C}'$ where only $a$ and $b$ have mutual commitment to each other (middle). Impure simplicial model $\mathcal{C}''$ where each agent has a pairwise commitment with one another, but no joint commitment to the group (right).}
    \label{fig:EX1}
\end{figure}

\begin{figure}[t]
    \centering
    \resizebox{13.5cm}{!}{%
\begin{tikzpicture}
\begin{scope}[xshift=-6cm,>=stealth, yshift=3cm]
    \node[circle,draw,minimum size=9mm, fill=red!80] (1) {$p_a$};
    \node[circle,draw,minimum size=9mm, fill=blue!50] (2) [right = 2cm  of 1]  {$p_b$};
    \node (g) [right = 1cm  of 1]  {};
    \node[circle,draw,minimum size=9mm, fill=green] (3) [above = 2cm  of g] {$p_c$}; 
    \node (h) [above = 0.5cm  of g]  {$X$};
    \node[circle,draw,minimum size=9mm, fill=blue!50] (4) [left = 2cm  of 1]  {$\neg p_b$};
    \node (f) [left = 0.8cm  of 1]  {};
    \node (h') [above = 0.5cm  of f]  {$Y$};
    \node[circle,draw,minimum size=9mm, fill=green] (5) [left = 2cm  of 3] {$\neg p_c$}; 
    \node (h'') [above = 2cm  of 4]  {\Large{$\mathcal{C}$}};
    
    \path[]
        (1) edge node[above] {} (2)
        (1) edge node[above] {} (3)
        (1) edge node[above] {} (4)
        (1) edge node[above] {} (5)
        (5) edge node[above] {} (4)
        (3) edge node[above] {} (2);

    \begin{pgfonlayer}{background}
\draw[fill=black!10,line width=0.1mm,line cap=round,line join=round] (1.center)--(2.center)--(3.center)--cycle;
\draw[fill=black!10,line width=0.1mm,line cap=round,line join=round] (5.center)--(4.center)--(1.center)--cycle;
\end{pgfonlayer}

    \node[circle,draw,minimum size=9mm, fill=red!80] (6) [right = 4cm  of 2]{$p_a$};
    \node[circle,draw,minimum size=9mm, fill=blue!50] (7) [right = 2cm  of 6]  {$p_b$};
    \node (g) [right = 1cm  of 6]  {};
    \node[circle,draw,minimum size=9mm, fill=green] (8) [above = 2cm  of g] {$p_c$}; 
    \node (h) [above = 0.5cm  of g]  {$X'$};
    \node[circle,draw,minimum size=9mm, fill=blue!50] (9) [left = 2cm  of 6]  {$\neg p_b$};
    \node (f) [left = 0.8cm  of 6]  {};
    \node[circle,draw,minimum size=9mm, fill=green] (10) [below = 2cm  of f] {$\neg p_c$};
    \node (h') [above = 0.2cm  of f]  {$W'$};
    \node (h'1) [below = 0.4cm  of f]  {$Y'$};
    \node (h'') [above = 2cm  of 9]  {\Large{$\mathcal{C}'$}};
    
    \path[]
        (6) edge node[above] {} (7)
        (6) edge node[above] {} (8)
        (6) edge node[above] {} (9)
        (8) edge node[above] {} (7)
        (6) edge node[above] {} (10)
        (9) edge node[above] {} (10)
        ; 

        \begin{pgfonlayer}{background}
\draw[fill=black!10,line width=0.1mm,line cap=round,line join=round] (6.center)--(7.center)--(8.center)--cycle;
\draw[fill=black!10,line width=0.1mm,line cap=round,line join=round] (6.center)--(9.center)--(10.center)--cycle;
\end{pgfonlayer}

 \node[circle,draw,minimum size=9mm, fill=red!80] (1a) [right = 2cm  of 7] {$p_a$};
    \node[circle,draw,minimum size=9mm, fill=blue!50] (2a) [right = 2cm  of 1a]  {$p_b$};
    \node (ga) [right = 1cm  of 1a]  {};
    \node[circle,draw,minimum size=9mm, fill=green] (3a) [above = 2cm  of ga] {$p_c$}; 
    \node (ha) [above = 0.5cm  of ga]  {$X''$};
    \node (g'a) [left = 1.1cm  of 1a]  {};
    \node[circle,draw,minimum size=9mm, fill=blue!50] (4a) [above = 1.2cm  of g'a]  {$\neg p_b$};
    \node (fa) [above = 1cm  of 2a]  {\Large{$\mathcal{C}''$}};
    \node (h'a) [above = 0.6cm  of 1a]  {$Y''$};
    
    \path[]
        (1a) edge node[above] {} (2a)
        (1a) edge node[above] {} (3a)
        (1a) edge node[above] {} (4a)
        (3a) edge node[above] {} (2a)
        (4a) edge node[above] {} (3a);

    \begin{pgfonlayer}{background}
\draw[fill=black!10,line width=0.1mm,line cap=round,line join=round] (1a.center)--(2a.center)--(3a.center)--cycle;
\draw[fill=black!10,line width=0.1mm,line cap=round,line join=round] (1a.center)--(4a.center)--(3a.center)--cycle;
\end{pgfonlayer}

    \end{scope}
\end{tikzpicture}

    }
    \caption{Simplicial models representing different examples from \autoref{ex:meeting_formal}, from left to right: versions 4, 5 and 6.}
    \label{fig:EX6}
\end{figure}
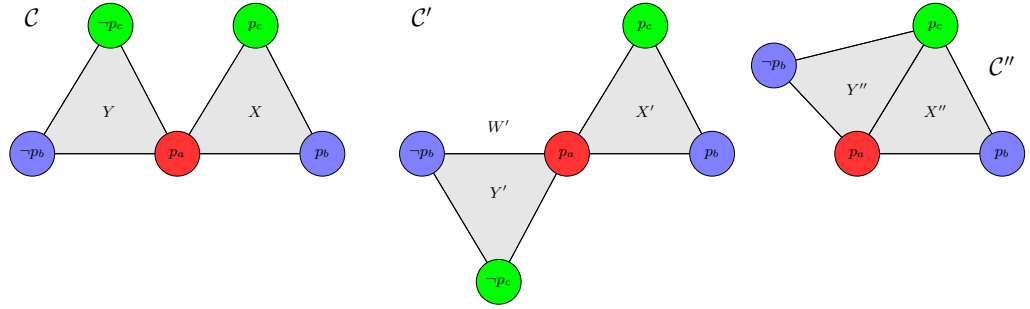

The above examples showed that our framework can model mutual commitment (and its absence) using impure simplicial complexes, even for very nuanced cases.

In the following section we prove soundness and completeness results of our logic via the canonical simplicial model construction \cite{rojo}.

\section{Soundness and Completeness of DSL}
\label{sec:SC1}

We begin by providing the following axiomatization for the \textbf{DSL} logic. For any group $G$, we first define the \emph{presence formula} $E_G := \mathbf{D}_G\top$; since $\top$ holds everywhere, semantically $\mathcal{C},X\models E_G \iff G\subseteq\chi(X)$.

\begin{definition}[Axiom system DSL]
\label{def:AX}
For $\varnothing \neq G, H, K \subseteq \agents$, and formulas restricted to the appropriate stratified fragment:

\begin{itemize}
    \item \textbf{Taut}: Propositional tautologies over $\mathcal{L}^{DSL}$ (with $\top$ primitive)
    \item \textbf{K}: $\mathbf{D}_G (\varphi \rightarrow \psi) \rightarrow (\mathbf{D}_G \varphi \rightarrow \mathbf{D}_G \psi)$, for $\varphi,\psi \in \mathcal{L}^{DSL}_G$
    \item \textbf{T}: $\mathbf{D}_G \varphi \rightarrow \varphi$, for $\varphi \in \mathcal{L}^{DSL}_G$
    \item \textbf{4}: $\mathbf{D}_G \varphi \rightarrow \mathbf{D}_G \mathbf{D}_G \varphi$, for $\varphi \in \mathcal{L}^{DSL}_G$
    \item \textbf{WN}: $p_a \rightarrow \mathbf{D}_{\{a\}} p_a$
    \item \textbf{C}: $\mathbf{D}_G \varphi \wedge \mathbf{D}_G \psi \rightarrow \mathbf{D}_G (\varphi \wedge \psi)$, for $\varphi,\psi \in \mathcal{L}^{DSL}_G$
    \item \textbf{RE}: $E_G \wedge \neg \mathbf{D}_G \varphi \rightarrow \mathbf{D}_G \neg \mathbf{D}_G \varphi$, for $\varphi \in \mathcal{L}^{DSL}_G$
    \item \textbf{EM}: $E_G \rightarrow E_H$, for $H\subseteq G$
    \item \textbf{MON}: $E_G \wedge \mathbf{D}_H \varphi \rightarrow \mathbf{D}_G \varphi$, for $H\subseteq G$, $\varphi\in\mathcal{L}^{DSL}_H$
    \item \textbf{EU}: $E_H \wedge E_K \rightarrow E_{H\cup K}$
    \item \textbf{NE}: $\bigvee_{a\in\agents}E_a$
    \item \textbf{MP}: from $\varphi$ and $\varphi \rightarrow \psi$ infer $\psi$
    \item \textbf{RM}: from $\vdash \varphi \rightarrow \psi$ (with $\varphi,\psi \in \mathcal{L}^{DSL}_G$) infer $\vdash \mathbf{D}_G \varphi \rightarrow \mathbf{D}_G \psi$
\end{itemize}

We write $DSL \vdash \varphi$ to indicate that $\varphi \in \mathcal{L}^{DSL}$ is derivable in DSL.
\end{definition}

A few comments on the axioms are in order.

In classical deontic logic, axiom \textbf{T} is usually rejected because the operator $O$ formalizes what ought to be the case, independently of what actually is. By contrast, $\mathbf{D}_G\varphi$ formalizes what a group has committed itself to realizing: the simplicial model represents configurations in which agents have collectively structured their choices so that $\varphi$ holds. In this sense, our modality is closer to the STIT achievement operator cstit \cite{Horty2001}, which also satisfies \textbf{T}. In our setting, axiom \textbf{T} mirrors the fact that agents are \textit{faithful} to their commitments, in the sense that they hold to the intention of fulfilling it and they will see into it that the state of affairs to which they are committed holds.\footnote{We plan to address unfaithful agents in future work.} In other words, our logic describes state of affairs that ought to hold because agents share a mutual commitment to achieve that state. Importantly, this does not prevent commitments to be modified or broken, as we will see in the dynamic case (\autoref{sec:DDSL}).

One would expect axiom \textbf{5} to hold, as it does for the epistemic interpretation of our operator. However, in our deontic setting it does not hold. As a counterexample, consider \autoref{fig:EX6} (center). $\mathcal{C}',W'\not \models p_c$ as $c \notin \chi(W')$. Thus, $\mathcal{C}',W'\models \neg \mathbf{D}_{\{a,b,c\}}p_c$. By \textbf{5} it should hold that $\mathcal{C}',W'\models \mathbf{D}_{\{a,b,c\}} \neg \mathbf{D}_{\{a,b,c\}}p_c$. However, $c \not\in \chi(W')$ makes the formula false. Instead, we adopt the \emph{relativized} axiom \textbf{RE}, guarded by presence: this restores a sound form of negative introspection, but only among faces where the group $G$ is genuinely present, exactly avoiding the counterexample above (where $c\notin\chi(W')$ blocks the antecedent $E_{\{a,b,c\}}$).

Furthermore, usually group modalities of the `distributed knowledge' kind, like our joint commitment modality, are characterized by a monotonicity axiom (for $G \subseteq  G'$, $\mathbf{D}_G \varphi \rightarrow \mathbf{D}_{G'} \varphi$). However, this does not hold in impure settings. As a counterexample, consider \autoref{fig:EX6} (middle). $\mathcal{C}', W' \models \mathbf{D}_{\{a,b\}} (p_a \wedge \neg p_b)$, but $\mathcal{C}', W' \not\models \mathbf{D}_{\{a,b,c\}} (p_a \wedge \neg p_b)$, as $c \notin \chi(W')$. In its place, axiom \textbf{MON} recovers monotonicity in a relativized form, guarded by the presence of the larger group.

Necessitation ($\varphi \rightarrow \mathbf{D}_G \varphi$) does not hold for arbitrary formulas $\varphi$ in case an agent in $G$ is not present in the face of evaluation. As a counterexample, consider \autoref{fig:EX1} (middle): $\mathcal{C}', X' \models \neg p_c$ since $c \notin \chi(X')$. By the same token, $\mathcal{C}', X' \not\models \mathbf{D}_{\{c\}} \neg p_c$, making $\neg p_c \rightarrow \mathbf{D}_{\{c\}} \neg p_c$ false. Necessitation holds in a weaker version only for atoms (\textbf{WN}), and for literals holds only if the agent is present in the face of evaluation. Axiom \textbf{C} and inference rule \textbf{RM} act as a replacement for the missing necessitation-derived steps for arbitrary formulas: they only ever produce a new $\mathbf{D}_G$-fact from an already-established $\mathbf{D}_G$-fact (\textbf{C}) or an already-established implication used to move between two conditional consequents (\textbf{RM}). Neither manufactures a $\mathbf{D}_G$-fact out of nothing at a configuration where the gate $G\subseteq\chi(X)$ might fail.

Axioms \textbf{EM}, \textbf{MON}, and \textbf{EU} govern how presence and commitment interact across different groups: \textbf{EM} lets presence of a larger group entail presence of any subgroup; \textbf{MON} lets an established subgroup commitment be relabeled at the level of any larger, currently-present group; \textbf{EU} lets presence of two groups \emph{at the same point} combine into presence of their union.
Finally, \textbf{NE} is a purely structural axiom, independent of the deontic reading: since every simplex is by definition a \emph{non-empty} set of vertices (\autoref{def:SimMod}), some agent is always present at any configuration. This fact is not derivable from the other axioms, all of which are conditional on already-established $\mathbf{D}_G$/$E_G$-facts, and its absence would make certain configurations of the canonical model ill-formed, as discussed after \autoref{lem:correctness}.

\begin{theorem}[Soundness]
The axiom system DSL (\autoref{def:AX}) is sound.
\end{theorem}

\begin{proof}
We prove each item separately.

\begin{itemize}

    \item \textbf{K:} Towards a contradiction suppose that $\mathcal{C},X \not \models \mathbf{D}_G (\varphi \rightarrow \psi) \rightarrow (\mathbf{D}_G \varphi \rightarrow \mathbf{D}_G \psi)$, meaning that $\mathcal{C},X  \models \mathbf{D}_G \varphi$ and $\mathcal{C},X \not \models \mathbf{D}_G \psi$ while $\mathcal{C},X \models \mathbf{D}_G (\varphi \rightarrow \psi)$. By the latter $\mathcal{C},Y \models (\varphi \rightarrow \psi)$ for all $Y \in C \text{ s.t. } G \subseteq \chi (X \cap Y)$, making $\mathbf{D}_G \psi$ true.

    \item \textbf{T:} Assume $\mathcal{C},X  \models \mathbf{D}_G \varphi$ then $\mathcal{C},Y \models \varphi$ for all $Y \in C \text{ s.t. } G \subseteq \chi (X \cap Y)$, including $X$.

    \item \textbf{4:} Assume $\mathcal{C},X \models  \mathbf{D}_G \varphi$. Then $\mathcal{C},Y \models \varphi$ for all $Y \in C$ s.t.\ $G \subseteq \chi (X \cap Y)$. Take an arbitrary face $Z$ such that $G \subseteq \chi (Y \cap Z)$. Then $\mathcal{C},Z \models \varphi$ and, since each simplex has at most one vertex of each color, $G \subseteq \chi (X \cap Z)$, meaning that $\mathcal{C},Y \models \mathbf{D}_G \varphi$. Since $\mathcal{C},Y \models \mathbf{D}_G\varphi$ holds for every such $Y$ and $X$ is among them, $\mathcal{C},X \models \mathbf{D}_G\mathbf{D}_G\varphi$ follows by the semantics.

    \item \textbf{WN:} Assume $ \mathcal{C}, X\models p_a$. Then there is a vertex $v_a \in X$ with $p_a \in l(v_a)$. For every $Y$ such that $a \in \chi (X \cap Y)$, $v_a \in Y$. Thus, $p_a \in l(Y)$, making $\mathcal{C}, Y \models p_a$. Therefore $\mathcal{C}, X \models \mathbf{D}_{\{a\}} p_a$.

    \item \textbf{C}: Suppose $\mathcal{C},X\models\mathbf{D}_G\varphi$ and $\mathcal{C},X\models\mathbf{D}_G\psi$. Both give $G\subseteq\chi(X)$, and for every $Y$ with $G\subseteq\chi(X\cap Y)$: $\mathcal{C},Y\models\varphi$ and $\mathcal{C},Y\models\psi$, hence $\mathcal{C},Y\models\varphi\wedge\psi$. So $\mathcal{C},X\models\mathbf{D}_G(\varphi\wedge\psi)$.

    \item \textbf{RM}: Suppose $\vdash\varphi\to\psi$, so $\models\varphi\to\psi$. Take any $\mathcal{C},X$ with $\mathcal{C},X\models\mathbf{D}_G\varphi$: then $G \subseteq\chi(X)$, and for every $Y$ with $G\subseteq\chi(X\cap Y)$, $\mathcal{C},Y\models\varphi$, hence $\mathcal{C},Y\models\psi$. The same gate condition, already established, gives $\mathcal{C},X\models\mathbf{D}_G\psi$. So $\mathcal{C},X\models\mathbf{D}_G\varphi\to\mathbf{D}_G\psi$.

    \item \textbf{RE}: Suppose $\mathcal{C},X\models E_G$ (so $G\subseteq\chi(X)$) and $\mathcal{C},X\models\neg\mathbf{D}_G\varphi$. Then there is $Y$ with $G\subseteq\chi(X\cap Y)$ and $\mathcal{C},Y\not\models\varphi$. We show $\mathcal{C},X\models\mathbf{D}_G\neg\mathbf{D}_G\varphi$: the gate already holds, so take an arbitrary $Z$ with $G\subseteq\chi(X\cap Z)$. By \autoref{lem:chromtrans} (chaining through $X$, using $G\subseteq\chi(Y\cap X)=\chi(X\cap Y)$ and $G\subseteq\chi(X\cap Z)$): $G\subseteq\chi(Y\cap Z)$. Combined with $\mathcal{C},Y\not\models\varphi$, $Y$ witnesses $\mathcal{C},Z\not\models\mathbf{D}_G\varphi$. Since $Z$ was arbitrary, $\mathcal{C},X\models\mathbf{D}_G\neg\mathbf{D}_G\varphi$.

    \item \textbf{EM}: $\mathcal{C},X\models E_G \Rightarrow G\subseteq\chi(X)$; since $H\subseteq G$, $H\subseteq\chi(X)$, i.e.\ $\mathcal{C},X\models E_H$.

    \item \textbf{MON}: Assume $\mathcal{C},X\models E_G$ (so $G\subseteq\chi(X)$) and $\mathcal{C},X\models \mathbf{D}_H\varphi$. Take any $Y$ with $G\subseteq\chi(X\cap Y)$; since $H\subseteq G$, $H\subseteq\chi(X\cap Y)$ too, so $\mathcal{C},Y\models\varphi$. As the gate $G\subseteq\chi(X)$ already holds, $\mathcal{C},X\models\mathbf{D}_G\varphi$.

    \item \textbf{EU}: $\mathcal{C},X\models E_H\wedge E_K \Rightarrow H\subseteq\chi(X)$ and $K\subseteq\chi(X)$, both at $X$. Hence $H\cup K\subseteq\chi(X)$, i.e.\ $\mathcal{C},X\models E_{H\cup K}$.

    \item \textbf{NE}: Let $\mathcal{C},X$ be arbitrary. By \autoref{def:SimMod}, $X$ is a non-empty finite subset of $V$, so it contains some vertex $v$; let $a:=\chi(v)$. Then $a\in\chi(X)$, so $\mathcal{C},X\models E_a$, hence $\mathcal{C},X\models\bigvee_{a'\in\agents}E_{a'}$.

\end{itemize}
\end{proof}

We can now move to the completeness proof, which uses the canonical simplicial model construction, a technique introduced in \cite{rojo} for a three-valued semantics, that we adapt here to a two-valued one.

The core challenge is to build the right canonical simplicial model. To identify the set of formulae belonging to a vertex in the canonical simplicial model construction, we use our novel modality limited to singleton groups $\mathbf{D}_{\{a\}}\varphi$. In particular, each $a$-colored vertex in the canonical simplicial model is a set $\mathbf{D}_{\{a\}}\Gamma$ of formulae of the form $\mathbf{D}_{\{a\}}\varphi$ already belonging to $\Gamma$. Each face of the canonical simplicial model is a finite set $X = \{\mathbf{D}_{\{a_1\}}\Gamma, \mathbf{D}_{\{a_2\}}\Gamma, \ldots , \mathbf{D}_{\{a_n\}}\Gamma\}$ of such sets for agents $a_1, a_2, \ldots, a_n$ present in configuration $X$, while the whole canonical simplicial model is a set of these finite sets $X$ (together with the valuation and coloring functions).

The Truth Lemma below has two directions with genuinely different characters. The $(\Rightarrow)$ direction, showing that $\mathbf D_G\psi\in\Gamma$ forces truth at \emph{every} face sharing $G$'s vertices with $X_\Gamma$  is settled directly by Group Invariance (\autoref{lem:groupinv}) once we know $\mathcal C^c$ is a genuine simplicial model (\autoref{lem:correctness}): no witness needs to be built, since Group Invariance already forces agreement across every such face simultaneously. The $(\Leftarrow)$ direction is existential: we must \emph{construct} one face where $\psi$ fails, using \autoref{lem:existenceG} (Group Existence) to obtain a consistent extension via Lindenbaum, and \autoref{cor:groupsymm} (Group vertex-sharing) to confirm this new face genuinely shares $G$'s vertices with $X_\Gamma$. \autoref{lem:presence} is the small bridging fact both directions need, relating syntactic presence ($E_G\in\Gamma$) to geometric presence ($G\subseteq\chi^c(X_\Gamma)$).

\begin{definition}[Maximal Consistent sets]
A set $\Gamma$ of formulae is consistent iff $\Gamma \not \vdash \bot$. Otherwise it is inconsistent. A set $\Gamma$ of formulae is maximal consistent if $\Gamma$ is consistent but no proper superset $\Gamma_0 \supset \Gamma$ is consistent.
\end{definition}

\begin{lemma}[Properties of maximal consistent sets]
\label{lem:propMCS}
Let $\Gamma, \Delta \in \mathcal{G}$ and $a \in \agents$:

\begin{enumerate}
    \item $\Gamma \vdash \varphi \Rightarrow \varphi \in \Gamma$

    \item $\varphi \in \Gamma \Leftrightarrow \neg \varphi \notin \Gamma$

    \item $(\varphi \wedge \psi) \in \Gamma \Leftrightarrow \varphi \in \Gamma \text{ and } \psi \in \Gamma$

    \item $p_a \in \Gamma$ iff $\mathbf{D}_{\{a\}}p_a \in \Gamma$

    \item $\mathbf{D}_{\{a\}}\varphi \in \Gamma \Rightarrow \neg \varphi \notin \Gamma$

    \item If $\mathbf{D}_{\{a\}}\Gamma \subseteq \Delta$ and $\mathbf{D}_{\{a\}}\Gamma \neq \varnothing$, then $\mathbf{D}_{\{a\}}\Delta \subseteq \Gamma$.

\end{enumerate}
\end{lemma}

\begin{proof}
We prove each item separately:

\begin{enumerate}
    \item Since $\Gamma \vdash \varphi$, we have $\vdash \bigvee \Gamma_0 \rightarrow \varphi$ for some finite $\varnothing \not= \Gamma_0 \subseteq \Gamma$. From tautology $\vdash (\bigvee \Gamma_0 \rightarrow \varphi) \rightarrow \neg (\bigvee \Gamma_0 \wedge \neg \varphi)$ by \textbf{MP} it follows $\vdash \neg (\bigvee \Gamma_0 \wedge \neg \varphi)$, making $\varphi \in \Gamma$.

    \item It follows from tautology $\vdash \neg (\varphi \wedge \neg \varphi)$ and $\Gamma$ being maximally consistent.

    \item It follows from previous items.

    \item From left to right it follows from axiom \textbf{WN}, from right to left it follows from axiom \textbf{T}.

    \item By axiom \textbf{T} and \textbf{MP}, $\mathbf{D}_{\{a\}}\varphi \in \Gamma$ gives $\varphi \in \Gamma$. By item 2, $\neg\varphi \notin \Gamma$.

    \item Since $\mathbf{D}_{\{a\}}\Gamma \neq \varnothing$, fix some $\mathbf{D}_{\{a\}}\chi_0 \in \Gamma$. By \textbf{RM} on $\vdash \chi_0 \to \top$ and \textbf{MP}, $\mathbf{D}_{\{a\}}\top = E_{\{a\}} \in \Gamma$.

    Let $\mathbf{D}_{\{a\}}\theta \in \Delta$; we show $\theta \in \Gamma$. Towards a contradiction, suppose $\neg\theta \in \Gamma$. By axiom \textbf{T}, its contrapositive $\neg\theta \to \neg\mathbf{D}_{\{a\}}\theta$ holds. \textbf{MP} gives $\neg\mathbf{D}_{\{a\}}\theta \in \Gamma$.

    By axiom \textbf{RE} with $\varphi := \theta$: $E_{\{a\}} \wedge \neg\mathbf{D}_{\{a\}}\theta \to \mathbf{D}_{\{a\}}\neg\mathbf{D}_{\{a\}}\theta$. Since both conjuncts of the antecedent are in $\Gamma$, deductive closure and \textbf{MP} give $\mathbf{D}_{\{a\}}\neg\mathbf{D}_{\{a\}}\theta \in \Gamma$.

    Since $\mathbf{D}_{\{a\}}\Gamma \subseteq \Delta$ and $\mathbf{D}_{\{a\}}\neg\mathbf{D}_{\{a\}}\theta$ is of the form $\mathbf{D}_{\{a\}}(\cdot)$ and belongs to $\Gamma$, it belongs to $\mathbf{D}_{\{a\}}\Gamma$, hence in $\Delta$: $\neg\mathbf{D}_{\{a\}}\theta \in \Delta$, contradicting $\mathbf{D}_{\{a\}}\theta\in\Delta$ and the consistency of $\Delta$. Hence $\theta \in \Gamma$.
\end{enumerate}
\end{proof}

\begin{corollary}
\label{cor:vertexeq}
If $\mathbf{D}_{\{a\}}\Gamma \subseteq \Delta$ and $\mathbf{D}_{\{a\}}\Gamma \neq \varnothing$, then $\mathbf{D}_{\{a\}}\Gamma = \mathbf{D}_{\{a\}}\Delta$.
\end{corollary}
\begin{proof}
$(\subseteq)$: every $\mathbf{D}_{\{a\}}\varphi\in\mathbf{D}_{\{a\}}\Gamma$ is in $\Delta$ by hypothesis, so it is also in $\mathbf{D}_{\{a\}}\Delta$.
$(\supseteq)$: by \autoref{lem:propMCS}(6), $\mathbf{D}_{\{a\}}\Delta\subseteq\Gamma$, so it belongs to $\mathbf{D}_{\{a\}}\Gamma$.
\end{proof}

\begin{lemma}[Lindenbaum Lemma \cite{ditmarsch2007dynamic}]
\label{lem:lind}
Any consistent set $\Gamma$ can be extended to a maximal consistent set $\Delta \supseteq \Gamma$.
\end{lemma}

\begin{definition}[Canonical simplicial model]
\label{def:CanSimMod}
Let $\mathcal{G}$ be the set of all maximal consistent sets. For $\Gamma \in \mathcal{G}$ and $a \in \agents$, define the \emph{projection}
\begin{equation}
    \mathbf{D}_{\{a\}}\Gamma := \Gamma \cap \{\mathbf{D}_{\{a\}}\varphi \mid \varphi \in \mathcal{L}^{DSL}\},
\end{equation}
i.e., the set of $\mathbf{D}_{\{a\}}$-formulas already belonging to $\Gamma$. We define
\begin{equation}
    X_\Gamma := \{\mathbf{D}_{\{a\}}\Gamma \mid a \in \agents,\ \mathbf{D}_{\{a\}}\Gamma \neq \varnothing\}.
\end{equation}
The vertex set of the canonical model is
\begin{equation}
    V^c := \{\mathbf{D}_{\{a\}}\Gamma \mid \Gamma \in \mathcal{G},\ a \in \agents,\ \mathbf{D}_{\{a\}}\Gamma \neq \varnothing\},
\end{equation}
where two pairs $(\Gamma,a)$, $(\Delta,a)$ yield the \emph{same} vertex object whenever $\mathbf{D}_{\{a\}}\Gamma = \mathbf{D}_{\{a\}}\Delta$ as sets.

The canonical simplicial model $\mathcal{C}^c = (C^c, \chi^c, l^c)$ is defined by:
\begin{itemize}
    \item $C^c = \{X \mid (\exists \Gamma \in \mathcal{G})(\varnothing \neq X \subseteq X_\Gamma)\}$;
    \item For $v = \mathbf{D}_{\{a\}}\Gamma \in V^c$: $\chi^c(v) = a$;
    \item For $v = \mathbf{D}_{\{a\}}\Gamma \in V^c$: $l^c(v) = \{p_a \in P_a \mid \mathbf{D}_{\{a\}}p_a \in \mathbf{D}_{\{a\}}\Gamma\}$.
\end{itemize}
\end{definition}

The main idea is to use non-empty sets $\mathbf{D}_{\{a\}} \Gamma$ as vertices of the canonical model, while sets $X_\Gamma$ are its faces. The coloring of a $\mathbf{D}_{\{a\}} \Gamma$ vertex is $a$, and the formulas holding at that vertex are those to which that agent is committed to. The projection allows empty vertices (when $a$ has no commitment content in $\Gamma$), which is exactly what impurity requires.

\begin{lemma}[Correctness]
\label{lem:correctness}
The object $\mathcal{C}^c$ of \autoref{def:CanSimMod} is a simplicial model, and $X_\Gamma\ne\varnothing$ for \emph{every} $\Gamma\in\mathcal{G}$.
\end{lemma}

\begin{proof}
\textbf{Non-emptiness of every face $X_\Gamma$.} By axiom \textbf{NE}, $\bigvee_{a\in\agents}E_a$ is a theorem, hence lies in every $\Gamma\in\mathcal{G}$ by \autoref{lem:propMCS}(1). By the standard disjunction property of maximal consistent sets (following from items 2–3 applied to the De Morgan expansion of $\vee$), some disjunct $E_{a_0}\in\Gamma$ for some $a_0\in\agents$. Since $E_{a_0}=\mathbf{D}_{\{a_0\}}\top$ is itself of the shape $\mathbf{D}_{\{a_0\}}(\cdot)$ and lies in $\Gamma$, $\mathbf{D}_{\{a_0\}}\Gamma\ni E_{a_0}$, so $\mathbf{D}_{\{a_0\}}\Gamma\ne\varnothing$, giving $X_\Gamma\ne\varnothing\in C^c$. 
This holds uniformly for every $\Gamma\in\mathcal{G}$.

\textbf{Downward closure and non-emptiness of faces.} By construction, $C^c$ consists precisely of all non-empty subsets of the various $X_\Gamma$. If $X \in C^c$ witnessed by $\Gamma$ and $\varnothing \neq Y \subseteq X$, then $Y \subseteq X \subseteq X_\Gamma$, so $Y \in C^c$.

\textbf{Finiteness.} For any $X \in C^c$, $X \subseteq X_\Gamma$ for some $\Gamma$, and $X_\Gamma$ has at most one element per agent, so $|X_\Gamma| \leq |\agents| = n+1$, hence $|X| \leq n+1$.

\textbf{Chromatic map well-defined.} Let $v = \mathbf{D}_{\{a\}}\Gamma \in V^c$ be non-empty. Every element of $\mathbf{D}_{\{a\}}\Gamma$ has the syntactic shape $\mathbf{D}_{\{a\}}\varphi$; by unique parsing of $\mathcal{L}^{DSL}$, $a$ is recoverable from any single element of $v$, so $\chi^c(v):=a$ is well-defined. Since $X_\Gamma$ contains at most one non-empty $\mathbf{D}_{\{a\}}\Gamma$ per agent, every face has at most one vertex of each color.

\textbf{Valuation.} For $v = \mathbf{D}_{\{a\}}\Gamma \in V^c$, $l^c(v) = \{p_a \in P_a \mid \mathbf{D}_{\{a\}}p_a \in \mathbf{D}_{\{a\}}\Gamma\} \subseteq P_a$ by construction.
\end{proof}

\begin{lemma}[Presence characterization]
\label{lem:presence}
For $\Gamma\in\mathcal G$ and $\varnothing\ne G\subseteq\agents$: $E_G\in\Gamma \iff G\subseteq\chi^c(X_\Gamma)$.
\end{lemma}
\begin{proof}
$(\Rightarrow)$ For each $a\in G$: \textbf{EM} gives $\vdash E_G\to E_{\{a\}}$; \textbf{MP} gives $E_{\{a\}}\in\Gamma$, so $\mathbf D_{\{a\}}\Gamma\ne\varnothing$, so $a\in\chi^c(X_\Gamma)$. As $a\in G$ was arbitrary: $G\subseteq\chi^c(X_\Gamma)$.

$(\Leftarrow)$ For each $a\in G$: $a\in\chi^c(X_\Gamma)$ means $\mathbf D_{\{a\}}\Gamma\ne\varnothing$, i.e.\ some $\mathbf D_{\{a\}}\theta_a\in\Gamma$; by \textbf{RM} and \textbf{MP}: $E_{\{a\}}\in\Gamma$. Induct on $|G|$: if $|G|=1$, this is $E_G$ itself. If $G=G_0\cup\{a\}$ with $|G_0|=|G|-1$: by IH, $E_{G_0}\in\Gamma$; with $E_{\{a\}}\in\Gamma$, \textbf{EU} and \textbf{MP} give $E_{G_0\cup\{a\}}=E_G\in\Gamma$.
\end{proof}

\begin{lemma}[Group Existence]
\label{lem:existenceG}
If $E_G\in\Gamma$ and $\neg\mathbf D_G\psi\in\Gamma$ (with $\psi\in\mathcal L^{DSL}_G$), then $\mathbf D_G\Gamma\cup\{\neg\psi\}$ is consistent, where $\mathbf D_G\Gamma:=\bigcup_{a\in G}\mathbf D_{\{a\}}\Gamma$.
\end{lemma}
\begin{proof}
Suppose not. By finiteness, there are $\mathbf D_{\{a_1\}}\varphi_1,\ldots,\mathbf D_{\{a_k\}}\varphi_k\in\Gamma$ ($k\ge0$, each $a_j\in G$, $\varphi_j\in\mathcal L^{DSL}_{\{a_j\}}$, drawn from $\mathbf D_G\Gamma$) with
\[
\vdash(\mathbf D_{\{a_1\}}\varphi_1\wedge\cdots\wedge\mathbf D_{\{a_k\}}\varphi_k)\to\psi.
\]

\textit{Case $k=0$.} Then $\vdash\psi$; by \textbf{Taut}+\textbf{MP}, $\vdash\top\to\psi$. Since $\top,\psi\in\mathcal L^{DSL}_G$, \textbf{RM} gives $\vdash E_G\to\mathbf D_G\psi$; \textbf{MP} with $E_G\in\Gamma$ gives $\mathbf D_G\psi\in\Gamma$ directly.

\textit{Case $k\ge1$.} For each $j$: writing $\theta_j:=\mathbf D_{\{a_j\}}\varphi_j\in\mathcal L^{DSL}_{\{a_j\}}$, by \textbf{4} and \textbf{MP}: $\mathbf D_{\{a_j\}}\theta_j\in\Gamma$. By \textbf{MON} (with $H=\{a_j\}\subseteq G$, using $E_G\in\Gamma$) and \textbf{MP}: $\mathbf D_G\theta_j\in\Gamma$.

By \textbf{C} chained across $j=1,\ldots,k$ (valid since each $\theta_j\in\mathcal L^{DSL}_{\{a_j\}}\subseteq\mathcal L^{DSL}_G$ by \autoref{lem:strat-mono}): $\mathbf D_G(\theta_1\wedge\cdots\wedge\theta_k)\in\Gamma$.

By \textbf{RM} on $\vdash(\theta_1\wedge\cdots\wedge\theta_k)\to\psi$: $\vdash\mathbf D_G(\theta_1\wedge\cdots\wedge\theta_k)\to\mathbf D_G\psi$; \textbf{MP} gives $\mathbf D_G\psi\in\Gamma$.

Either way $\mathbf D_G\psi\in\Gamma$, contradicting $\neg\mathbf D_G\psi\in\Gamma$ (\autoref{lem:propMCS}(2)).
\end{proof}

\begin{corollary}[Existence, singleton case]
\label{lem:existence}
If $\mathbf D_{\{a\}}\Gamma\ne\varnothing$ and $\neg\mathbf D_{\{a\}}\psi\in\Gamma$, then $\mathbf D_{\{a\}}\Gamma\cup\{\neg\psi\}$ is consistent.
\end{corollary}
\begin{proof}
Since $\mathbf D_{\{a\}}\Gamma\ne\varnothing$, some $\mathbf D_{\{a\}}\theta_0\in\Gamma$; by \textbf{RM} and \textbf{MP}: $E_{\{a\}}\in\Gamma$. The claim is now \autoref{lem:existenceG} instantiated at $G=\{a\}$, noting $\mathbf D_{\{a\}}\Gamma=\bigcup_{a'\in\{a\}}\mathbf D_{\{a'\}}\Gamma$. 
\end{proof}

\begin{corollary}[Group vertex-sharing]
\label{cor:groupsymm}
If $\mathbf D_G\Gamma\subseteq\Delta$ and $E_G\in\Gamma$, then $\mathbf D_{\{a\}}\Gamma=\mathbf D_{\{a\}}\Delta$ for every $a\in G$.
\end{corollary}
\begin{proof}
Fix $a\in G$. By \textbf{EM} and \textbf{MP}: $E_{\{a\}}\in\Gamma$, so $\mathbf D_{\{a\}}\Gamma\ni E_{\{a\}}\neq\varnothing$. Since $\mathbf D_{\{a\}}\Gamma\subseteq\mathbf D_G\Gamma\subseteq\Delta$, \autoref{cor:vertexeq} gives $\mathbf D_{\{a\}}\Gamma=\mathbf D_{\{a\}}\Delta$. As $a\in G$ was arbitrary, this holds for every $a\in G$.
\end{proof}

\begin{lemma}[Truth Lemma]
\label{lem:truth}
For every maximal consistent set $\Gamma \in \mathcal{G}$ and every formula $\varphi \in \mathcal{L}^{DSL}$:
$$\varphi \in \Gamma \;\iff\; \mathcal{C}^c, X_\Gamma \models \varphi.$$
\end{lemma}

\begin{proof}
By induction on the construction of $\varphi$.

\paragraph{Case $\varphi = p_a$.}

$(\Rightarrow)$ Assume $p_a \in \Gamma$. By \autoref{lem:propMCS}(4), $\mathbf{D}_{\{a\}}p_a \in \Gamma$, so $\mathbf{D}_{\{a\}}p_a \in \mathbf{D}_{\{a\}}\Gamma$ and $\mathbf{D}_{\{a\}}\Gamma \neq \varnothing$. Hence $v_a = \mathbf{D}_{\{a\}}\Gamma \in X_\Gamma$ and $a \in \chi^c(X_\Gamma)$. By the definition of $l^c$, $p_a \in l^c(v_a) \subseteq l^c(X_\Gamma)$. Therefore $\mathcal{C}^c, X_\Gamma \models p_a$.

$(\Leftarrow)$ Assume $\mathcal{C}^c, X_\Gamma \models p_a$. Then $a \in \chi^c(X_\Gamma)$, so $v_a = \mathbf{D}_{\{a\}}\Gamma \in X_\Gamma$, and $p_a \in l^c(X_\Gamma)=l^c(v_a)$. By definition of $l^c$, $\mathbf{D}_{\{a\}}p_a \in \mathbf{D}_{\{a\}}\Gamma \subseteq \Gamma$. By axiom \textbf{T} and \textbf{MP}, $p_a \in \Gamma$.

\paragraph{Case $\varphi = \neg\psi$.}

$(\Rightarrow)$ Assume $\neg\psi \in \Gamma$. By \autoref{lem:propMCS}(2), $\psi \notin \Gamma$. By the induction hypothesis (contrapositive of $(\Leftarrow)$ for $\psi$), $\mathcal{C}^c, X_\Gamma \not\models \psi$, hence $\mathcal{C}^c, X_\Gamma \models \neg\psi$.

$(\Leftarrow)$ Assume $\mathcal{C}^c, X_\Gamma \models \neg\psi$. By the induction hypothesis (contrapositive of $(\Rightarrow)$ for $\psi$), $\psi \notin \Gamma$. By \autoref{lem:propMCS}(2), $\neg\psi \in \Gamma$.

\paragraph{Case $\varphi = \psi_1 \wedge \psi_2$.}

$(\Rightarrow)$ Assume $(\psi_1 \wedge \psi_2) \in \Gamma$. By \autoref{lem:propMCS}(3), $\psi_1,\psi_2 \in \Gamma$. By the IH, $\mathcal{C}^c, X_\Gamma \models \psi_1$ and $\mathcal{C}^c, X_\Gamma \models \psi_2$.

$(\Leftarrow)$ Assume $\mathcal{C}^c, X_\Gamma \models \psi_1 \wedge \psi_2$. By the IH, $\psi_1,\psi_2 \in \Gamma$. By \autoref{lem:propMCS}(3), $(\psi_1 \wedge \psi_2) \in \Gamma$.

\paragraph{Case $\varphi = \mathbf{D}_G\psi$, for $\varnothing\ne G\subseteq\agents$, $\psi\in\mathcal L^{DSL}_G$.}

$(\Rightarrow)$ Assume $\mathbf D_G\psi\in\Gamma$.

\emph{Gate condition:} By \textbf{Taut}+\textbf{MP}, $\vdash\psi\to\top$; since $\psi,\top\in\mathcal L^{DSL}_G$, \textbf{RM} gives $\vdash\mathbf D_G\psi\to E_G$; \textbf{MP} gives $E_G\in\Gamma$. By \autoref{lem:presence}: $G\subseteq\chi^c(X_\Gamma)$.

\emph{Content at $X_\Gamma$:} By \textbf{T}+\textbf{MP}: $\psi\in\Gamma$. By the IH: $\mathcal C^c,X_\Gamma\models\psi$.

\emph{Propagation:} Let $W\in C^c$ be arbitrary with $G\subseteq\chi^c(X_\Gamma\cap W)$. Since $\mathcal C^c$ is a simplicial model (\autoref{lem:correctness}) and $\psi\in\mathcal L^{DSL}_G$, \autoref{lem:groupinv} applies to $X=X_\Gamma$, $X'=W$: $\mathcal C^c,X_\Gamma\models\psi \iff \mathcal C^c,W\models\psi$. Combined with the previous step: $\mathcal C^c,W\models\psi$.

Since $W$ was arbitrary $\mathcal C^c,X_\Gamma\models\mathbf D_G\psi$.

$(\Leftarrow)$, contrapositive: assume $\mathbf D_G\psi\notin\Gamma$; by \autoref{lem:propMCS}(2), $\neg\mathbf D_G\psi\in\Gamma$.

If $G\not\subseteq\chi^c(X_\Gamma)$: the gate already fails, so $\mathcal C^c,X_\Gamma\not\models\mathbf D_G\psi$.

If $G\subseteq\chi^c(X_\Gamma)$: by \autoref{lem:presence}, $E_G\in\Gamma$. By \autoref{lem:existenceG}, $\mathbf D_G\Gamma\cup\{\neg\psi\}$ is consistent; extend via \autoref{lem:lind} to $\Delta\in\mathcal G$. By \autoref{cor:groupsymm}: $\mathbf D_{\{a\}}\Gamma=\mathbf D_{\{a\}}\Delta$ for every $a\in G$, so $G\subseteq\chi^c(X_\Gamma\cap X_\Delta)$. Since $\neg\psi\in\Delta$, $\psi\notin\Delta$; by the IH (contrapositive), $\mathcal C^c,X_\Delta\not\models\psi$. So $X_\Delta$ witnesses $\mathcal C^c,X_\Gamma\not\models\mathbf D_G\psi$.
\end{proof}

\begin{theorem}[Completeness]
\label{theo:compl}
The axiomatization of \autoref{def:AX} is complete with respect to impure simplicial models: $\models \varphi \text{ implies } \vdash \varphi$
\end{theorem}

\begin{proof}
We prove the contrapositive. Suppose $\not \vdash \varphi$ for some formula $\varphi$. Then $\{\neg \varphi\}$ is a consistent set. By Lindenbaum Lemma (\autoref{lem:lind}), $\{\neg \varphi\}$ is a subset of some maximal consistent set $\Gamma$. By \autoref{lem:correctness}, $X_\Gamma\ne\varnothing$, so $X_\Gamma\in C^c$ is a genuine point of the canonical model. By \autoref{lem:truth}, $\mathcal{C}^c, X_\Gamma \models \neg \varphi$, making $\varphi$ not valid.
\end{proof}

\begin{remark}
Group modalities of the `distributed knowledge' kind are known to pose a genuine obstacle to completeness via the standard canonical-model technique: syntactic agreement on individual agents' commitments does not, in general, force semantic agreement on group-level content, since the group modality is by design not reducible to its members' contributions (\autoref{fig:bigtri}). 
In the epistemic simplicial setting, obtaining a complete axiomatization for a fully general distributed-knowledge operator has required moving beyond plain simplicial complexes to a richer structure \cite{ericDK}. The stratified grammar $\mathcal L^{DSL}_G$ together with \autoref{lem:groupinv} is what allows \textbf{DSL} to obtain full completeness for $\mathbf D_G$ at every group size, within plain simplicial complexes: 
restricting $\mathbf D_G\varphi$'s own content to $\varphi\in\mathcal L^{DSL}_G$ is precisely what lets Group Invariance settle the hard direction of the Truth Lemma directly, without needing the richer machinery the unrestricted case would demand.
\end{remark}

\section{Properties of Deontic Simplicial Models}
\label{sec:prop}

A key advantage of the simplicial structure is that it allows structural properties of the normative landscape to be read off directly from the geometry.
For example, checking whether a complex allows for a configuration where all agents share at least one joint commitment equals checking whether the simplicial model has at least a face of maximal dimension, as shown by the following theorem:

\begin{theorem}
A deontic simplicial model $\mathcal{C}=(C,\chi,l)$ allows for commitment formulas involving all agents iff $\exists X \in C$ s.t. $|X|= n+1$.

\end{theorem}

\begin{proof}

$(\Rightarrow)$ Consider $\mathcal{C}$ and $X \in C$ such that $\mathcal{C},X \models \mathbf{D}_\agents \varphi$, for some $\varphi \in \mathcal{L}^{DSL}$. Since $|\agents|=n+1$, $|X| = n+1$.

$(\Leftarrow)$ Consider $X \in C$ such that $|X|= n+1$. Since $|\agents|=n+1$, $|X| = |\agents|$ meaning that $\agents \subseteq \chi(X)$, making $\mathcal{C},X \models \mathbf{D}_\agents \varphi$ true for some formula $\varphi \in \mathcal{L}^{DSL}$.

\end{proof}

The above theorem makes sure that there is a simple way to check whether a complex, no matter its size, has a configuration where all agents are committed to some state of affairs. 
It can also be useful to focus on only on lower dimensions of joint commitments. The operation that allows this is the skeleton operation:

\begin{definition}[Skeleton]
\label{def:skel}
The set of simplices of a simplicial complex $\mathcal{C}$ of dimension at most $k$ is a subcomplex of $\mathcal{C}$, called the $k$-skeleton, denoted $skel_k(\mathcal{C})$.
\end{definition}

For example, the $0$-skeleton $skel_0(\mathcal{C})$ of
a complex $\mathcal{C}$ is the complex comprising all of its vertices, The $1$-skeleton is made up of all of its vertices and edges, and so on. An example is shown in \autoref{fig:skel}.

\begin{figure}[t]
    \centering
    \resizebox{11cm}{!}{%
\begin{tikzpicture}
\begin{scope}[xshift=-6cm,>=stealth, yshift=3cm]
 
 \node[circle,draw,minimum size=9mm, fill=red!80] (1a) [right = 2cm  of 7] {$p_a$};
    \node[circle,draw,minimum size=9mm, fill=blue!50] (2a) [right = 2cm  of 1a]  {$p_b$};
    \node (ga) [right = 1cm  of 1a]  {};
    \node[circle,draw,minimum size=9mm, fill=green] (3a) [above = 2cm  of ga] {$p_c$}; 
    \node (ha) [above = 0.5cm  of ga]  {$X$};
    \node (g'a) [left = .7cm  of 1a]  {};
    \node[circle,draw,minimum size=9mm, fill=blue!50] (4a) [above = 1.8cm  of g'a]  {$\neg p_b$};
    \node (fa) [above = 1cm  of 2a]  {\Large{$\mathcal{C}$}};
    \node (h'a) [above = 0.7cm  of 1a]  {$Y$};
    
    \path[]
        (1a) edge node[above] {} (2a)
        (1a) edge node[above] {} (3a)
        (1a) edge node[above] {} (4a)
        (3a) edge node[above] {} (2a)
        (4a) edge node[above] {} (3a);

    \begin{pgfonlayer}{background}
\draw[fill=black!10,line width=0.1mm,line cap=round,line join=round] (1a.center)--(2a.center)--(3a.center)--cycle;
\draw[fill=black!10,line width=0.1mm,line cap=round,line join=round] (1a.center)--(4a.center)--(3a.center)--cycle;
\end{pgfonlayer}

 \node[circle,draw,minimum size=9mm, fill=red!80] (1) [right = 2cm  of 2a] {$p_a$};
    \node[circle,draw,minimum size=9mm, fill=blue!50] (2) [right = 2cm  of 1]  {$p_b$};
    \node (g) [right = 1cm  of 1]  {};
    \node[circle,draw,minimum size=9mm, fill=green] (3) [above = 2cm  of g] {$p_c$}; 
    \node (g') [left = .7cm  of 1]  {};
    \node[circle,draw,minimum size=9mm, fill=blue!50] (4) [above = 1.8cm  of g']  {$\neg p_b$};
    \node (f) [above = 1cm  of 2]  {\Large{$skel_1(\mathcal{C})$}};

    \path[]
        (1) edge node[above] {} (2)
        (1) edge node[above] {} (3)
        (1) edge node[above] {} (4)
        (3) edge node[above] {} (2)
        (4) edge node[above] {} (3);

    \end{scope}
\end{tikzpicture}

    }
    \caption{The $1$-skel operation of complex $\mathcal{C}$.}
    \label{fig:skel}
\end{figure}
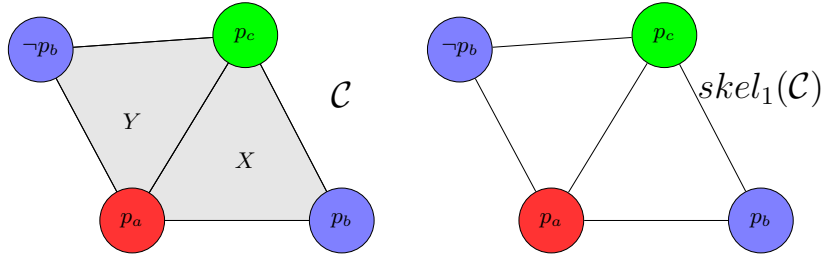

\begin{theorem}

Let $\mathcal{C}$ be a deontic simplicial complex, $k \leq n$ and $|G| \leq k+1$. For any atomic formula $p_a \in P_a$ and any face $X \in skel_k(\mathcal{C})$:

$$\mathcal{C}, X \models \mathbf{D}_G p_a \implies skel_k(\mathcal{C}), X \models \mathbf{D}_G p_a$$

Moreover, $skel_k(\mathcal{C})$ is the largest subcomplex of $\mathcal{C}$ in which all joint commitments of groups of size at most $k+1$ to atomic formulas are preserved.
\end{theorem}

\begin{proof}
Assume $|G| \leq k+1$ and $\mathcal{C}, X \models \mathbf{D}_G p_a$.
Take any $Y \in skel_k(C)$ with $G \subseteq \chi(X \cap Y)$. Since $skel_k(C) \subseteq C,  Y \in C$, and the assumption gives $\mathcal{C}, Y \models p_a$, meaning $a \in \chi(Y)$ and $p_a \in l(Y)$. Since the skeleton operation does not alter vertex valuations  and since $Y$ was arbitrary, $skel_k(\mathcal{C}), X \models \mathbf{D}_G p_a.$

The second statement follows directly from Definition \ref{def:skel}.
%Since $X \in skel_k(\mathcal{C})$, $dim(X) \leq dim(skel_k(\mathcal{C}))$ and since $k \leq n$. $dim(skel_k(\mathcal{C})) \leq dim(\mathcal{C})$,then $skel_k(\mathcal{C}), X \models \mathbf{D}_G \varphi$.

\end{proof}

Other powerful tools, expressing commitments (or lack thereof) from the perspective of agents (or group of agents), include the \textit{star} and the \textit{link} operations:

\begin{definition}[Star]
\label{def:star}
Let $\mathcal{C}$ be a simplicial complex. The star of a simplex $X \in \mathcal{C}$, written $Star(X,\mathcal{C})$ is the subcomplex of $\mathcal{C}$ whose facets are the maximal simplices of $\mathcal{C}$ that contains $X$.
\end{definition}

\begin{theorem}
Given a deontic simplicial complex $\mathcal{C}$ and a vertex $v_a \in C$, the faces of $Star(v_a,\mathcal{C})$ are exactly the simplices $Y \in C$ such that $Y \cup \{v_a\} \in C$.
\end{theorem}

\begin{proof}
($\subseteq$) If $Y \in Star(v_a,\mathcal{C})$, then by \autoref{def:star} $Y \subseteq F$ for some facet $F \in \mathcal{F}(C)$ with $v_a \in F$. Thus $Y \cup \{v_a\} \subseteq F \in C$, and since $C$ is downward closed, $Y \cup \{v_a\} \in C$.

($\supseteq$) If $Y \cup \{v_a\} \in C$, extend $Y \cup \{v_a\}$ to some facet $F \in \mathcal{F}(C)$. Then $v_a \in F$ and $Y \subseteq F$, so $Y$ is a face of a facet containing $v_a$, i.e. $Y \in Star(v_a,\mathcal{C})$.
\end{proof}

The informal reading is that $Y$ is in 
a's star exactly when $Y$'s configuration is extendable by a fresh commitment from a, that is, a could join whoever is already committed in $Y$.

Finally, the \textit{link} of an agent determines which commitments other agents would still have among themselves if that agent were to drop its commitment:

\begin{definition}[Link]
\label{def:link}
Let $\mathcal{C}$ be a simplicial complex. The link of a simplex $X \in \mathcal{C}$, written $Link(X,\mathcal{C})$ is the subcomplex of $\mathcal{C}$ consisting of all simplices in $Star(X,\mathcal{C})$ that do not have common vertices with $X$.
\end{definition}

\begin{theorem}
Given a deontic simplicial complex $\mathcal{C}$ and a vertex $v_a$, the faces of $Link(v_a,\mathcal{C})$ are exactly the simplices $Y \in C$ such that $v_a \notin Y$ and $Y \cup \{v_a\} \in C$.
\end{theorem}

\begin{proof}
By \autoref{def:link}, $Link(v_a,\mathcal{C}) = \{Y \in Star(v_a,\mathcal{C}) \mid v_a \notin Y\}$. Substituting the characterization of $Star(v_a,\mathcal{C})$ from the previous theorem yields exactly $\{Y \in C \mid v_a \notin Y \text{ and } Y \cup \{v_a\} \in C\}$.
\end{proof}

\begin{figure}[t]
    \centering
    \resizebox{8cm}{!}{%
\begin{tikzpicture}
\begin{scope}[xshift=-6cm,>=stealth, yshift=3cm]

  \coordinate (O)  at (0, 0);
  \coordinate (A)  at ($(O) + (90:2)$);
  \coordinate (B)  at ($(O) + (30:2)$);
  \coordinate (C)  at ($(O) + (-30:2)$);
  \coordinate (D)  at ($(O) + (-90:2)$);
  \coordinate (E)  at ($(O) + (-150:2)$);
  \coordinate (F)  at ($(O) + (150:2)$);

  % Fill triangle gray
  \foreach \from/\to in {A/B, B/C, C/D, D/E, E/F, F/A} {
    \filldraw[fill=gray!30, draw=black, thick] (O) -- (\from) -- (\to) -- cycle;
  }

  % Solid black dots 
  \foreach \p in {O, A, B, C, D, E, F} {
    \filldraw[black] (\p) circle (3pt);
  }

  \node[above right=2pt] at (O) {\Large{$v$}};

  \node at ($(O) + (0, -2.6)$) {(a) $Star(v,\mathcal{C})$};

  \coordinate (O2) at (5, 0);
  \coordinate (A2) at ($(O2) + (90:2)$);
  \coordinate (B2) at ($(O2) + (30:2)$);
  \coordinate (C2) at ($(O2) + (-30:2)$);
  \coordinate (D2) at ($(O2) + (-90:2)$);
  \coordinate (E2) at ($(O2) + (-150:2)$);
  \coordinate (F2) at ($(O2) + (150:2)$);

  % Outer border 
  \draw[thick] (A2) -- (B2) -- (C2) -- (D2) -- (E2) -- (F2) -- cycle;

  \foreach \p in {A2, B2, C2, D2, E2, F2} {
    \draw[dotted, thick] (O2) -- (\p);
  }

  \foreach \p in {A2, B2, C2, D2, E2, F2} {
    \filldraw[black] (\p) circle (3pt);
  }

  \draw[dotted, thick] (O2) circle (3pt);

  \node[above right=2pt] at (O2) {\Large{$v$}};

  \node at ($(O2) + (0, -2.6)$) {(b) $Link(v,\mathcal{C})$};

    \end{scope}
\end{tikzpicture}

    }
    \caption{Star operation and Link operation from vertex $v$.}
    \label{fig:starlink}
\end{figure}
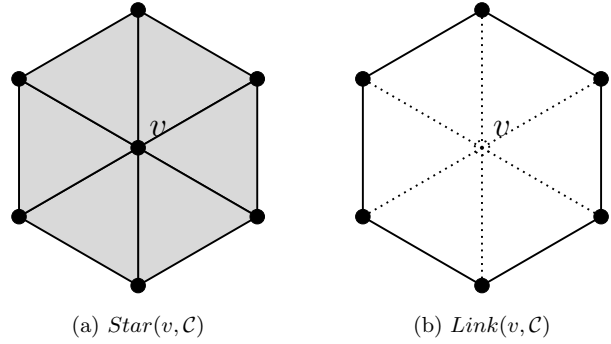

Analogous notions for group of agents can be formulated by using $Star(X,\mathcal{C})$ and $Link(X,\mathcal{C})$.

The link operation hints at a dynamic component in the `removal' of vertex $v_a$. 
In the following we take a different approach to commitment dynamics, adapting the tools from Dynamic Epistemic Logic \cite{ditmarsch2007dynamic} to simplicial complexes.

\section{Dynamic Deontic Logic for Simplicial Complexes}
\label{sec:DDSL}

In this section we introduce the dynamic deontic simplicial logic \textbf{DDSL}, extending \textbf{DSL} with an action modality that captures transformations of the simplicial structure, 
thereby simulating agents' choice among mutually exclusive (joint) commitments.

We start with some preliminary definitions, adapting the seminal work of \cite{know_in_simpl}.
We fix a set of actions $\{\alpha_a,\beta_a,\ldots\}$ indexed by agents. The agent-index indicates which agent is responsible for executing that action. 
Actions here are not to be intended like actual epistemic actions as in the DEL setting \cite{ditmarsch2007dynamic}. Rather, they represent statements of commitments from agents (e.g., in a contract). In our interpretation, actions do not alter the epistemic state of the world, they transform the current set of commitments of a set of agents into a new one, obtained by agents re-stating their commitments. Consequently, these newly stated commitments gain a geometric interpretation: a vertex represents the new commitment of a single agent, an edge a new joint commitment by two agents, and so on, and all commitments that are not re-iterated are dropped.

Adopting the action models machinery from DEL, we give them a deontic and simplicial interpretation using the simplicial commitment update models:

\begin{definition}[Simplicial commitment update model]
\label{def:commUmod}
A simplicial commitment update model $\mathcal{U}$ is a triple $(E, \chi, com)$
where $E \ne \varnothing$ is a simplicial complex, $\chi$ is a chromatic map, and $com$ is a commitment function $com : V(E) \mapsto \mathcal{L}^{DSL}_{\{a\}}$ assigning to each vertex $\alpha_a \in V(E)$ a formula from the static fragment restricted to agent $a$, where $a \in \chi(\alpha_a)$.

In particular, we require $com$ to satisfy the following stability condition:
if $\mathcal{C}, X \models com(\alpha_a)$ and $v_a \in Y \subseteq X$ then $\mathcal{C}, Y \models com(\alpha_a)$.
\end{definition}

In the same manner that standard action models in Dynamic Epistemic Logic \cite{ditmarsch2007dynamic} share structural similarities with Kripke models, simplicial commitment update models are also structurally similar to simplicial models: they both are made of a simplicial complex plus a chromatic map. What sets simplicial commitment update models apart is the commitment function $com$.
Formally, it has the same role of the \textit{precondition function} in Dynamic Epistemic Logic \cite{ditmarsch2007dynamic}, in the sense that it `gates' the formulas that are allowed to appear in the resulting model. However, the commitment function has a very distinct interpretation, as it represents the commitment of agents. An $a$-vertex in an update model with $p_a$ represents the fact that agent $a$ now commits to $p_a$. If an agent does not alter their commitments, the commitment function is simply $\top$.

The stability condition on $com$ requires that whenever an agent's commitment formula $com(\alpha_a)$ is satisfied at a face $X$, it continues to be satisfied at any subface $Y \subseteq X$ that still contains a's vertex $v_a$.
For boolean combinations of propositional variables from $P_a$, stability holds because the truth of $p_a$ depends only on the local valuation at vertex $v_a$, which is fixed regardless of the surrounding face. 
For singleton modalities 
$\mathbf{D}_{\{a\}}\psi$, stability holds because by chromaticity the set of faces checked when evaluating $\mathbf{D}_{\{a\}}\psi$ is the same at any subface $Y \subseteq X$ containing $v_a$ as at $X$ itself. Boolean combinations of these two forms therefore also satisfy the condition.

One clear advantage of using simplicial commitment update models is that they can naturally model joint actions, representing agents in $G$ simultaneously committing to certain formulas.
Formally, given a simplicial update model $\mathcal{U}=(E,\chi,com)$, we use notation $\alpha_G \subseteq E$ for joint actions made of individual actions $\alpha_a \in \alpha_G$. 
For joint actions $\alpha_G$, we use $com(\alpha_G) = \bigwedge_{a \in G} com (\alpha_a)$.
Being represented by vertices in 
$\mathcal{U}$, single agent actions have dimension $0$, while joint actions have dimension greater than $0$ in the corresponding simplicial commitment update model.
Examples of different simplicial commitment update models are offered in \autoref{fig:SAME}.

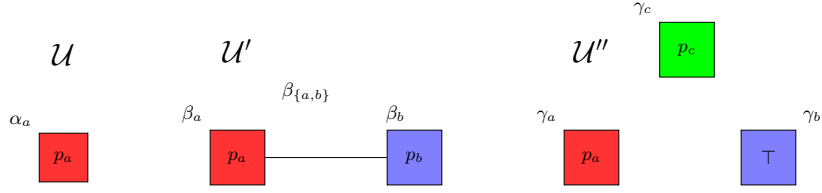
\begin{figure}[t]
    \centering
    \resizebox{11cm}{!}{%
\begin{tikzpicture}
    \begin{scope}[xshift=-6cm,>=stealth, yshift=3cm]
    
    \node[ label=135:$\alpha_a$, draw,minimum size=8mm,fill=red!80] (3) {$ p_a$};

    \node (h) [above = 1cm  of 3] {\Large{$\mathcal{U}$}};
    
    \node[label=135:$\beta_a$, draw,minimum size=9mm,fill=red!80] (1)  [right = 2cm  of 3] {$ p_a$};

    \node[label=95:$\beta_b$, draw,minimum size=9mm,fill=blue!50] (2) [right = 2cm  of 1] {$p_b$};

    \node (h) [above = 1cm  of 1] {\Large{$\mathcal{U'}$}};
    \node (h') [above = .5cm  of 1] {};
    \node (f) [right = .5cm  of h'] {$\beta_{\{a,b\}}$};
    
    \path[]
        (1) edge node[above] {} (2);

    \node[draw,minimum size=9mm, fill=red!80,label=135:$\gamma_a$] (4) [right = 2cm  of 2] {$ p_a$};
    \node[,draw,minimum size=9mm, fill=blue!50,label=45:$\gamma_b$] (5) [right = 2cm  of 4]  {$\top$};%{$p_b$};
    \node (g) [right = 1cm  of 4]  {};
    \node[draw,minimum size=9mm, fill=green,label=135:$\gamma_c$] (6) [above = 1.2cm  of g] {$p_c$}; 
    %\node (h) [above = 0.5cm  of g]  {$X$};
    \node (f) [above = 1cm  of 4]  {\Large{$\mathcal{U}''$}};
    
   % \path[]
   %     (4) edge node[above] {} (5)
   %     (4) edge node[above] {} (6)
   %     (6) edge node[above] {} (5);

    \end{scope}
\end{tikzpicture}
    }
    \caption[Examples of simplicial commitment update models.]{Examples of simplicial commitment update models. $\mathcal{U}$ (left) represents the single agent action where $a$ commits to $p_a$. $\mathcal{U}'$ (middle) illustrates the joint action $\beta_{\{a,b\}}$ made of single actions $\beta_a$ and $\beta_b$ where $a$ and $b$ jointly commit to $p$. %respectively $p_a$ and $p_b$ .
    $\mathcal{U}''$ (right) represents the simultaneous (but not joint) action where $a$ and $c$ commit to $p$ while $b$ does not alter her commitments ($\top$), but with no joint commitment.}
    \label{fig:SAME}
\end{figure}

The effects of re-instating joint commitments are captured by the simplicial product update operation:

\begin{definition}[Simplicial product update]
\label{def:SPU}
Given a simplicial model $\mathcal{C}=(C, \chi, l)$ with simplices $X \subseteq C$ and a simplicial commitment update model $\mathcal{U}=(E, \chi', com)$ with simplices $\alpha_G \subseteq E$, their simplicial product update results in the updated simplicial model $\mathcal{C}\otimes \mathcal{U}=(C^u, \chi^u, l^u)$, defined as follows:

We define the set of vertices of $\mathcal{C}\otimes \mathcal{U}$ as follows: $V^u := \{(v_a,\alpha_a)\in V(C) \times V(E) \mid \chi(v_a)=\chi'(\alpha_a)\}$, and

%for non-empty $Z \subseteq V^u$

\begin{itemize}
    \item $C^u    := \{(X,\alpha_G) \subseteq C \times E | \mathcal{C},X \models com(\alpha_a) \text{ for all } \alpha_a \in \alpha_G \text{ and s.t. } \chi(X)=\chi'(\alpha_G) \}$, where each pair $(X, \alpha_G)$
 represents the set of vertex pairs $\{(v_a, \alpha_a) \in V^u \mid v_a \in X,\ \alpha_a \in \alpha_G\}$,

    \item $\chi^u(X,\alpha_G) = \chi(X)=\chi'(\alpha_G)$
 
    %$\chi^u(v,\alpha_a) = \chi(v) = \chi'(\alpha_a) $
    \item $p_a \in l^u(v,\alpha_a)$ iff $p_a \in l(v) $ for all $p_a \in P_a$ 
    %and $\chi^u(v,\alpha_a) = \chi(v)=\chi'(\alpha_a)$
    %$p_a \in l^u(v,\alpha_a)$ iff $p_a \in l(v) $ for all $p_a \in P_a$
\end{itemize}
    
\end{definition}

The stability condition of $com$ (\autoref{def:commUmod}) guarantees that $C^u$ is a simplicial complex, i.e., closed under taking non-empty subsets. 
Concretely, if $(X, \alpha_G) \in C^u$ and $(Y, \beta_G)$ is obtained by removing some agents then $Y \in C$ and $\beta_G \in E$ by downward closure of $C$ and $E$ respectively, and the commitment condition at $Y$ follows from stability. Intuitively, the condition expresses that a commitment statement by agent $a$ concerns only $a$'s own local state at vertex $v_a$, and does not depend on which other agents happen to be present in the same configuration.

Intuitively, the simplicial product update keeps exactly those faces of $\mathcal C$ whose agents' new commitments (as specified by $com$) are already satisfied, paired with the corresponding action vertices; agents whose new commitment fails are simply excluded from the resulting configuration. Valuation and coloring on the surviving vertices are otherwise unchanged.

\begin{definition}[Language $\mathcal{L}^{DDSL}$]
\label{def:langLDDSL}
For $G \subseteq \agents$, the fragment $\mathcal{L}^{DDSL}_G$ is defined by mutual recursion on $G$:

\begin{align*}
\varphi ::= & \top \ \mid\ p_a\ (a\in G) \ \mid\ \neg\varphi \ \mid\ (\varphi\wedge\varphi) \ \mid\ \mathbf{D}_H\varphi\ (\varnothing\ne H\subseteq G,\ \varphi\in\mathcal{L}^{DDSL}_H) \ \mid\  \\
      & [\alpha_H]\varphi\ (\varnothing\ne H\subseteq G,\ \varphi\in\mathcal{L}^{DDSL}_H) 
\end{align*}

%\[
%\varphi ::= \top \ \mid\ p_a\ (a\in G) \ \mid\ \neg\varphi \ \mid\ (\varphi\wedge\varphi) \ \mid\ \mathbf{D}_H\varphi\ (\varnothing\ne H\subseteq G,\ \varphi\in\mathcal{L}^{DDSL}_H) \ \mid\ 
%[\alpha_H]\varphi\ (\varnothing\ne H\subseteq G,\ \varphi\in\mathcal{L}^{DDSL}_H),
%\]
where in the last clause $\alpha_H$ ranges over $H$-colored faces of a fixed, already-given simplicial commitment update model $\mathcal U=(E,\chi',com)$ with $com(\gamma_a)\in\mathcal L^{DDSL}_{\{a\}}$ for every $\gamma_a\in V(E)$. We set $\bot:=\neg\top$, define $\to,\vee$ as usual, and $\langle\alpha_H\rangle\varphi:=\neg[\alpha_H]\neg\varphi$. The full language is $\mathcal{L}^{DDSL}:=\mathcal{L}^{DDSL}_\agents$.

\end{definition}

\begin{remark}
Since $\mathcal L^{DSL}_{\{a\}}\subseteq\mathcal L^{DDSL}_{\{a\}}$, the codomain $\mathcal L^{DSL}_{\{a\}}$ of $com$ required by \autoref{def:commUmod} for primitive update models already satisfies the requirement above. As we will see, composed update models (\autoref{def:compo}) are assigned the richer codomain $\mathcal L^{DDSL}_{\{a\}}$ directly.
\end{remark}

\begin{lemma}[Monotonicity of stratification, DDSL]
\label{lem:strat-mono-ddsl}
If $H\subseteq G$, then $\mathcal L^{DDSL}_H\subseteq\mathcal L^{DDSL}_G$.
\end{lemma}
\begin{proof}
As \autoref{lem:strat-mono}, with an identical additional case for $\varphi=[\alpha_K]\psi$ ($K\subseteq H$, $\psi\in\mathcal L^{DDSL}_K$): then $K\subseteq H\subseteq G$, so $\varphi\in\mathcal L^{DDSL}_G$ via the same clause with the same $K$.
\end{proof}

\begin{remark}
At $G=\{a\}$, since the only nonempty $H\subseteq\{a\}$ is $H=\{a\}$ itself, \autoref{def:langLDDSL} specializes to
\[
\varphi ::= \top \mid p_a \mid \neg\varphi \mid (\varphi\wedge\varphi) \mid \mathbf{D}_{\{a\}}\varphi \mid [\gamma_a]\varphi,
\]
the single-agent dynamic fragment used throughout this section. Here $\gamma_a$ always belongs to a fixed, previously-constructed update model, so $com(\gamma_a)$ is simply whatever formula that model already assigns. 
\end{remark}

\begin{lemma}[Restriction]
\label{lem:restriction}
Let $X\in C$ and $\varnothing\ne H\subseteq\chi(X)$. Define $X_H:=\{v\in X:\chi(v)\in H\}$. Then $X_H\in C$ and $\chi(X_H)=H$.
\end{lemma}
\begin{proof}
By construction $\chi(X_H)\subseteq H$. Since $H\subseteq\chi(X)$, for each $a\in H$ there is a (unique, by chromaticity) vertex $v_a\in X$ with $\chi(v_a)=a$, and $v_a\in X_H$; so $\chi(X_H)\supseteq H$, giving $\chi(X_H)=H$. In particular $X_H\ne\varnothing$. Since $C$ is downward closed and $\varnothing\ne X_H\subseteq X\in C$, $X_H\in C$.
\end{proof}

\begin{definition}[Semantics]
\label{def:DDSLsem}
The semantics of $\mathcal{L}^{DDSL}$ extends the semantics of $\mathcal{L}^{DSL}$ (\autoref{def:semL}), with the following additional clause. For $\varnothing\ne H\subseteq\agents$, $\alpha_H\in V(E)$ an $H$-colored face of a fixed update model $\mathcal U=(E,\chi',com)$, and any $X\in C$:
\[
\mathcal{C},X \models [\alpha_H]\varphi \;\iff\;
\big(\mathcal{C},X\models E_H \text{ and } \mathcal{C},X_H\models com(\alpha_H)\big) \implies \mathcal{C}\otimes\mathcal{U},(X_H,\alpha_H)\models\varphi,
\]
where $X_H$ is given by \autoref{lem:restriction} whenever $\mathcal C,X\models E_H$ (i.e.\ $H\subseteq\chi(X)$); if $\mathcal C,X\not\models E_H$, the antecedent is false and the clause holds vacuously. The dual is $\langle\alpha_H\rangle\varphi:=\neg[\alpha_H]\neg\varphi$.We sometimes use the shorthand $\pi(\alpha_H) := E_H\wedge com(\alpha_H)$ for the joint gate-and-commitment precondition.
\end{definition}

Intuitively, applying action $\alpha_G$ results in $\varphi$ holding in the resulting simplicial model $\mathcal{C}\otimes\mathcal{U} $, provided $\alpha_G$ is executable in $X$ i.e., the formulas of the commitment function of each individual action in the joint action hold. 

The requirements of the gate condition $E_H$ and of the restriction to $X_H$ are necessary, because two things can go wrong when a subgroup $H$ acts at a configuration $X$ that may involve more agents than $H$. 
First, at $X$ some member of $H$ may be missing, in which case there is nothing for the action to act on, and the clause is vacuously true (via $E_H$). 
Second, even when $H$ is present, the action's effects should only ever touch $H$'s own vertices, not the unrelated commitments of other agents who happen to share the configuration. This is taken care of by restricting to $X_H$ (\autoref{lem:restriction}). The Examples below only ever need the trivial case $H=\chi(X)$, but the general clause is essential once a subgroup's action must be evaluated at a strictly larger, mixed configuration.

\begin{remark}
At $H=\chi(X)$, $E_H$ holds trivially and $X_H=X$, recovering the naive matching-group clause used implicitly throughout the Examples below. At $H=\{a\}$, $X_{\{a\}}=\{v_a\}$, recovering the single-agent case previously stated as a separate local clause; that definition is superseded and is not restated separately.
\end{remark}

\begin{lemma}
\label{lem:diamond}
For $\varnothing\ne H\subseteq\agents$, $\alpha_H\in V(E)$, $X\in C$, and $\theta\in\mathcal{L}^{DDSL}$:
\[
\mathcal{C},X\models\langle\alpha_H\rangle\theta \;\iff\;
\mathcal{C},X\models E_H \text{ and } \mathcal{C},X_H\models com(\alpha_H) \text{ and } \mathcal{C}\otimes\mathcal{U},(X_H,\alpha_H)\models\theta.
\]
\end{lemma}
\begin{proof}
$\langle\alpha_H\rangle\theta:=\neg[\alpha_H]\neg\theta$. By \autoref{def:DDSLsem}, $\mathcal C,X\models[\alpha_H]\neg\theta$ iff $(\mathcal C,X\models E_H$ and $\mathcal C,X_H\models com(\alpha_H))$ implies $\mathcal C\otimes\mathcal U,(X_H,\alpha_H)\not\models\theta$. Negating, using $\neg(P\to\neg Q)\equiv P\wedge Q$, gives the claim.
\end{proof}

%In particular, 
%for single-agent actions, $\gamma_a\in V(E), 
%\mathcal{C},X \models [\gamma_a]\varphi \;\Longleftrightarrow\;
%\mathcal{C},\{v_a\}\models com(\gamma_a) \text{ implies } \mathcal{C}\otimes\mathcal{U},(v_a,\gamma_a)\models\varphi$,
%and the semantics for the single-agent action diamond operation is: 
%$\mathcal{C},X\models\langle\gamma_a\rangle\theta \;\Longleftrightarrow\;
%\mathcal{C},\{v_a\}\models com(\gamma_a) \text{ and } \mathcal{C}\otimes\mathcal{U},(v_a,\gamma_a)\models\theta.$

\subsection{\textit{Examples of Commitment Dynamics}}

Our goal is to use actions to model agent choices when facing multiple possible (mutually exclusive) commitments.
We explore the intended meaning of our semantics using some examples.

\begin{example}[Simple choice]
\label{ex:Schoice}
For this example we use Version 6 of our motivating example. The simplicial complex is reported again here in
\autoref{fig:EX3} (left). As stated already in \autoref{ex:meeting_formal}, $a$ and $c$ have a mutual commitment to $p$, formally, $\mathcal{C}, X \models \mathbf{D}_{\{a,c\}} (p_a \wedge p_c)$. $b$ is not yet committed to only one of the two options, since, formally, $\mathcal{C}, X \models \mathbf{D}_{\{a,b,c\}} (p_a \wedge p_b \wedge p_c)$ and $\mathcal{C}, Y\models \mathbf{D}_{\{a,b,c\}} (p_a \wedge \neg p_b \wedge p_c)$. 
Say $b$ decides to join the party $p$. This joint action $\alpha_G$ is represented in \autoref{fig:EX3} (middle). In it, $b$ commits to $p_b$, while $a$ and $c$ do not alter their commitments, as highlighted by the use of $\top$ in the preconditions of the action model.
The resulting model is shown in \autoref{fig:EX3} (right). The joint decision results in a mutual commitment to $p$: $\mathcal{C},X \models [\alpha_G] \mathbf{D}_{\{a,b,c\}} (p_a \wedge p_b \wedge p_c)$, as $\mathcal{C},X \models com(\alpha_G)$ and $\mathcal{C}\otimes \mathcal{U},(X,\alpha_G)\models \mathbf{D}_{\{a,b,c\}} (p_a \wedge p_b \wedge p_c)$.

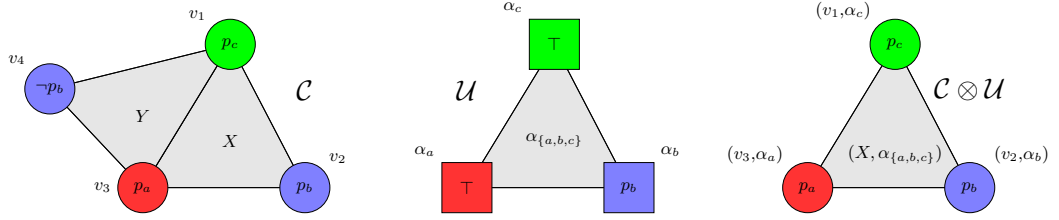
\begin{figure}[t]
    \centering
    \resizebox{14cm}{!}{%
\begin{tikzpicture}
    
    \node[circle,draw,minimum size=9mm, fill=red!80,label=180:$v_3$] (1) {$p_a$};
    \node[circle,draw,minimum size=9mm, fill=blue!50,label=45:$v_2$] (2) [right = 2cm  of 1]  {$p_b$};
    \node (g) [right = 1cm  of 1]  {};
    \node[circle,draw,minimum size=9mm, fill=green,label=135:$v_1$] (3) [above = 2cm  of g] {$p_c$}; 
    \node (h) [above = 0.5cm  of g]  {$X$};
    \node (g') [left = 1.1cm  of 1]  {};
    \node[circle,draw,minimum size=9mm, fill=blue!50,label=135:$v_4$] (4) [above = 1.2cm  of g']  {$\neg p_b$};
    \node (f) [above = 1cm  of 2]  {\Large{$\mathcal{C}$}};
    \node (h') [above = 0.6cm  of 1]  {$Y$};
    
    \path[]
        (1) edge node[above] {} (2)
        (1) edge node[above] {} (3)
        (1) edge node[above] {} (4)
        (3) edge node[above] {} (2)
        (4) edge node[above] {} (3);

    \begin{pgfonlayer}{background}
\draw[fill=black!10,line width=0.1mm,line cap=round,line join=round] (1.center)--(2.center)--(3.center)--cycle;
\draw[fill=black!10,line width=0.1mm,line cap=round,line join=round] (1.center)--(4.center)--(3.center)--cycle;
\end{pgfonlayer}

    \node[draw,minimum size=9mm, fill=red!80,label=135:$\alpha_a$] (14) [right = 2cm  of 2] {$\top$};
    \node[,draw,minimum size=9mm, fill=blue!50,label=45:$\alpha_b$] (15) [right = 2cm  of 14]  {$p_b$};
    \node (g) [right = 1cm  of 14]  {};
    \node[draw,minimum size=9mm, fill=green,label=135:$\alpha_c$] (16) [above = 2cm  of g] {$\top$}; 
    \node (h) [above = 0.5cm  of g]  {$\alpha_{\{a,b,c\}}$};
    \node (f) [above = 1cm  of 14]  {\Large{$\mathcal{U}$}}; %\alpha_{\{a,b,c\}}$};
    
    \path[]
        (14) edge node[above] {} (15)
        (14) edge node[above] {} (16)
        (16) edge node[above] {} (15);

     \begin{pgfonlayer}{background}
\draw[fill=black!10,line width=0.1mm,line cap=round,line join=round] (15.center)--(14.center)--(16.center)--cycle;
\end{pgfonlayer}

    \node[circle,draw,minimum size=9mm, fill=red!80, label=135:$(v_3{,}\alpha_a)$] (6) [right = 2.3cm  of 15] {$p_a$};
    \node[circle,draw,minimum size=9mm, fill=blue!50, label=45:$(v_2{,}\alpha_b)$] (7) [right = 2cm  of 6]  {$p_b$};
    \node (g) [right = 1cm  of 6]  {};
    \node[circle, draw,minimum size=9mm, fill=green,, label=135:$(v_1{,}\alpha_c)$] (8) [above = 2cm  of g] {$p_c$}; 
    \node (h) [above = 0.2cm  of g]  {$(X,\alpha_{\{a,b,c\}})$};
    \node (f) [above = 1cm  of 7]  {\Large{$\mathcal{C}\otimes\mathcal{U}$}};
    
    \path[]
        (6) edge node[above] {} (7)
        (6) edge node[above] {} (8)
        (8) edge node[above] {} (7);

    \begin{pgfonlayer}{background}
\draw[fill=black!10,line width=0.1mm,line cap=round,line join=round] (6.center)--(7.center)--(8.center)--cycle;
\end{pgfonlayer}

\end{tikzpicture}
    }
    \caption[Simplicial models for \autoref{ex:Schoice}.]{Simplicial models for \autoref{ex:Schoice}. Initial simplicial model $\mathcal{C}$ (left) 
    with two possible configurations, $X$ (everybody commits to $p$) and $Y$ ($a$ and $c$ commit to $p$, $b$ to $\neg p$).
    simplicial commitment update model $\mathcal{U}$ (middle), where $b$ commits to $p_b$ and $a$ and $c$ do not alter their commitments.
    Resulting simplicial model $\mathcal{C}\otimes\mathcal{U}$ (right), where all agents are committed to do $p$.}
    \label{fig:EX3}
\end{figure}

\end{example}

\begin{example}[Partial disagreement]
\label{ex:disa}
Take again the same scenario as above, where $a$ and $c$ are committed to join the party, while $b$ still has not decided, but is committed to either come or not. The initial situation is identical to the one in the previous example, depicted in \autoref{fig:EX3} (left).
$b$ suggests a theme for the party that $a$ and $c$ really do not like. After some discussion, the three friends end up arguing and decide separately what to do. $a$ and $c$ jointly decide to go to the party, but they do not communicate their decision to $b$, sharing no commitment to him. 
At the same time, $b$ decides not to go to the party, again without communicating her commitment to the others, violating her previous joint commitment.
The corresponding impure simplicial commitment update model is shown in \autoref{fig:disa} (left).
The resulting simplicial model \autoref{fig:disa} (right) is hence partially disconnected, representing the fact that $a$ and $c$ have a mutual commitment among each other, but none towards $b$, and that $b$ has no commitment to the others.
In particular, $a$ and $c$ retain their mutual commitment: $\mathcal{C}, \overline{AC} \models [\beta_G] D_{\{a,c\}}(p_a \wedge p_c)$, while $b$ does not: $\mathcal{C}, X \not\models [\beta_G] D_{\{a,b,c\}}(p_a \wedge p_c \wedge \neg p_b)$ because $b \notin \chi(X,\beta_G)$, as $(X,\beta_G) \notin C^u$. 
This is precisely an instance of axiom \textbf{R4b} (see \autoref{def:RAX}): here $H=\{a,c\}$ is the acting group, $K=\{a,b,c\}\not\subseteq H$, so $[\beta_G]\mathbf D_{\{a,b,c\}}(\cdots)$ can only be true if the action itself fails to execute, which it does not, so the formula is false.
Importantly, $\mathcal{C \otimes \mathcal{U}'}$ still presents some commitment for agent $b$, namely in vertex $(v_4,\beta_b)$, although it is not a joint commitment: $\mathcal{C \otimes \mathcal{U}'},(v_4,\beta_b) \models \mathbf{D}_{\{b\}}\neg p_b$.

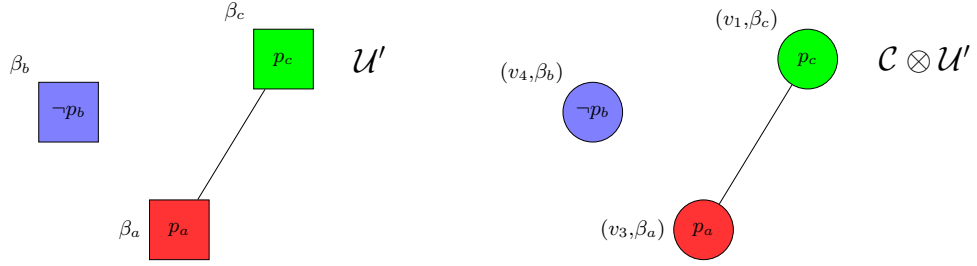
\begin{figure}[t]
    \centering
    \resizebox{13cm}{!}{%
\begin{tikzpicture}
    
    \node[draw,minimum size=9mm, fill=red!80,label=180:$\beta_a$] (1) {$p_a$};
    \node (g) [right = 1cm  of 1]  {};
    \node[draw,minimum size=9mm, fill=green,label=135:$\beta_c$] (3) [above = 2cm  of g] {$p_c$}; 
    \node (g') [left = 1.1cm  of 1]  {};
    \node[draw,minimum size=9mm, fill=blue!50,label=135:$\beta_b$] (4) [above = 1.2cm  of g']  {$\neg p_b$};
    \node (f) [right = .5cm  of 3] {\Large{$\mathcal{U}'$}};
    
    \path[]       
        (1) edge node[above] {} (3);

    \node[circle,draw,minimum size=9mm, fill=red!80,label=180:$(v_3{,}\beta_a)$] (10) [right = 7cm  of 1] {$p_a$};
    \node (a) [right = 1cm  of 10]  {};
    \node[circle,draw,minimum size=9mm, fill=green,label=135:$(v_1{,}\beta_c)$] (30) [above = 2cm  of a] {$p_c$}; 
    \node (a') [left = 1.1cm  of 10]  {};
    \node[circle,draw,minimum size=9mm, fill=blue!50,label=135:$(v_4{,}\beta_b)$] (40) [above = 1.2cm  of a']  {$\neg p_b$};
    \node (b) [right = .5cm  of 30]  {\Large{$\mathcal{C}\otimes \mathcal{U}'$}};
    
    \path[]       
        (10) edge node[above] {} (30);

\end{tikzpicture}
    }
    \caption[Simplicial models for \autoref{ex:disa}]{Simplicial commitment update model $\mathcal{U}'$ (left) representing $a$ and $c$'s joint commitment to $p$ and simultaneous but not joint commitment of $b$ to $\neg p$, without any commitment to the others. Resulting model $\mathcal{C}\otimes \mathcal{U}'$ (right) representing $a$ and $c$'s joint commitment to $p$ and $b$'s commitment to $\neg p$.}
    \label{fig:disa}
\end{figure}

\end{example}

\begin{example}[Ghosting]
\label{ex:ghost}
We use again the same starting scenario as above, depicted in \autoref{fig:EX3} (left). However, this time $b$ violates her previous joint commitment by simply not answering in the group chat.
The corresponding impure simplicial commitment update model is shown in \autoref{fig:ghost} (left). Differently from the previous example (\autoref{fig:disa} (left)), $\mathcal{U}''$ does not have any blue colored vertex, meaning that $b$ completely withdraws any commitment, while $a$ and $c$ retain any previous commitment (as indicated by the commitment to $\top$ in their joint action).
The resulting simplicial model \autoref{fig:ghost} (right) is made only of $a$ and $c$ vertices, linked by an edge.
Also in this example, $a$ and $c$ retain their mutual commitment: $\mathcal{C}, \overline{AC} \models [\gamma_G] D_{\{a,c\}}(p_a \wedge p_c)$.
Unlike the previous example, there is no simplex where $b$ has a commitment for herself.

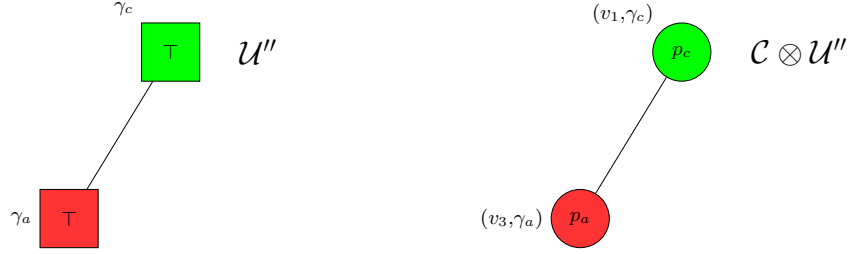
\begin{figure}[t]
    \centering
    \resizebox{12cm}{!}{%
\begin{tikzpicture}
    
    \node[draw,minimum size=9mm, fill=red!80,label=180:$\gamma_a$] (1) {$\top$};
    \node (g) [right = 1cm  of 1]  {};
    \node[draw,minimum size=9mm, fill=green,label=135:$\gamma_c$] (3) [above = 2cm  of g] {$\top$}; 
    \node (g') [left = 1.1cm  of 1]  {};
    \node[] (4) [above = 1.2cm  of g']  {};
    \node (f) [right = .5cm  of 3] {\Large{$\mathcal{U}''$}};
    
    \path[]       
        (1) edge node[above] {} (3);

    \node[circle,draw,minimum size=9mm, fill=red!80,label=180:$(v_3{,}\gamma_a)$] (10) [right = 7cm  of 1] {$p_a$};
    \node (a) [right = 1cm  of 10]  {};
    \node[circle,draw,minimum size=9mm, fill=green,label=135:$(v_1{,}\gamma_c)$] (30) [above = 2cm  of a] {$p_c$}; 
    \node (a') [left = 1.1cm  of 10]  {};
    \node[] (40) [above = 1.2cm  of a']  {};
    \node (b) [right = .5cm  of 30]  {\Large{$\mathcal{C}\otimes \mathcal{U}''$}};
    
    \path[]       
        (10) edge node[above] {} (30);

\end{tikzpicture}
    }
    \caption[Simplicial models for \autoref{ex:ghost}]{Simplicial commitment update model $\mathcal{U}''$ (left) representing $a$ and $c$'s joint commitment to $p$ and the absence of any commitment of $b$. Resulting model $\mathcal{C}\otimes \mathcal{U}''$ (right) representing $a$ and $c$'s joint commitment to $p$.}
    \label{fig:ghost}
\end{figure}

\end{example}

\section{Soundness and Completeness of DDSL}
\label{sec:SC2}

Before proving soundness and completeness, we introduce the definition of simplicial commitment update model composition and illustrate some of its properties:
%Since the simplicial product update operation is conducted on a vertex by vertex basis, each representing a commitment of an agent, it is important to highlight the properties that hold at vertices:

\begin{definition}[Simplicial commitment update model composition]
\label{def:compo}
Given $\mathcal{U}=(E,\chi,com)$, $\mathcal{U}'=(E',\chi',com')$, with $com,com' \in \mathcal{L}^{DDSL}_{\{a\}}$ (\autoref{def:langLDDSL} at $G=\{a\}$), 
their composition
$\mathcal{U};\mathcal{U}'=(E'',\chi'',com'')$ is defined as follows:

Its vertices are $V(E''):=\{\alpha_a;\beta_a : \alpha_a\in V(E),\beta_a\in V(E'),\chi(\alpha_a)=\chi'(\beta_a)\}$, and:

\begin{itemize}
    \item $E''=E\times E'$,
    \item $\chi''(\alpha_G;\beta_G)=\chi(\alpha_G)=\chi'(\beta_G)$ and
    \item $com''(\alpha_a;\beta_a) := \langle\alpha_a\rangle\, com'(\beta_a) \text{ for each vertex } \alpha_a;\beta_a \in V(E'')$.
\end{itemize}

\end{definition}

In particular, we use the standard shorthand $com''(\alpha_G;\beta_G):=\bigwedge_{a\in G}com''(\alpha_a;\beta_a)$.

Intuitively, $\mathcal U;\mathcal U'$ represents performing the commitments of $\mathcal U$ first, and then, from the resulting configuration, performing those of $\mathcal U'$.\footnote{Axiom \textbf{R5} will let us later replace any two consecutive actions by this single combined one.}
The following lemma extends the notion of group stability  (\autoref{lem:groupinv}) to the dynamic language:

\begin{lemma}[Group Stability of $\mathcal L^{DDSL}_H$]
\label{lem:groupstabDDSL}
Let $\mathcal C$ be any simplicial model, $\varnothing\ne H\subseteq\agents$, $X,Y\in C$ with $H\subseteq\chi(X\cap Y)$, and $\varphi\in\mathcal L^{DDSL}_H$. Then $\mathcal C,X\models\varphi\iff\mathcal C,Y\models\varphi$.
\end{lemma}
\begin{proof}
By structural induction on $\varphi$ (\autoref{def:langLDDSL} at level $H$).

\textit{$\varphi=\top$}: trivial.

\textit{$\varphi=p_a$} ($a\in H$): since $H\subseteq\chi(X\cap Y)$, in particular $a\in\chi(X\cap Y)$, so $X,Y$ share the unique $a$-vertex $v_a$; by disjointness of the $P_a$'s, $p_a\in l(X)\iff p_a\in l(v_a)\iff p_a\in l(Y)$.

\textit{$\varphi=\neg\psi$, $\varphi=\psi_1\wedge\psi_2$} ($\psi,\psi_1,\psi_2\in\mathcal L^{DDSL}_H$, by the grammar's own clauses): immediate from the IH, applied to the (possibly dynamic) subformula(s) directly.

\textit{$\varphi=\mathbf D_K\psi$} ($K\subseteq H$, $\psi\in\mathcal L^{DDSL}_K$): as before, $K\subseteq\chi(X)$ and $K\subseteq\chi(Y)$ both hold, and the witness sets $\{Z\in C:K\subseteq\chi(X\cap Z)\}$ and $\{Z\in C:K\subseteq\chi(Y\cap Z)\}$ coincide (via \autoref{lem:chromtrans}, chaining through $X$ and $Y$ respectively).

\textit{$\varphi=[\alpha_K]\psi$} ($K\subseteq H$, $\psi\in\mathcal L^{DDSL}_K$): as before, $K\subseteq\chi(X\cap Y)$ gives $X,Y$ sharing every $a$-vertex for $a\in K$, hence $X_K=Y_K$; and $K\subseteq\chi(X)$, $K\subseteq\chi(Y)$ both hold. By \autoref{def:DDSLsem}, $\mathcal C,X\models[\alpha_K]\psi$ and $\mathcal C,Y\models[\alpha_K]\psi$ are evaluated via the identical pair $(X_K,\alpha_K)=(Y_K,\alpha_K)$.
\end{proof}

The intuition behind this lemma is that an $H$-stratified formula can only "see" the $H$-colored part of a configuration. Consequently, if two configurations show that part identically, the formula can't tell them apart, no matter what else differs between them.

\begin{corollary}[Stability of $\mathcal L^{DDSL}_{\{a\}}$]
\label{lem:stabFrag}
For any simplicial model $\mathcal{C}$, any $X,Y\in C$ with $v_a\in Y\subseteq X$, and any $\varphi\in\mathcal{L}^{DDSL}_{\{a\}}$: $\mathcal{C},X\models\varphi \iff \mathcal{C},Y\models\varphi$.
\end{corollary}
\begin{proof}
Since $Y\subseteq X$ and both are simplices of the same chromatic complex, $v_a\in Y$ is the unique $a$-vertex of both $X$ and $Y$ (chromaticity forces at most one per color per simplex, and $Y\subseteq X$ inherits it), so $v_a\in X\cap Y$, i.e.\ $\{a\}\subseteq\chi(X\cap Y)$. The claim is \autoref{lem:groupstabDDSL} instantiated at $H=\{a\}$.
\end{proof}

The following lemma proves that the definition of the commitment function resulting from a simplicial commitment update model composition is correct:

\begin{lemma}[Composition equivalence]
\label{lem:comEquiv}
For $\varnothing\ne H\subseteq\agents$, $\alpha_H,\beta_H$ with $\chi(\alpha_H)=\chi'(\beta_H)=H$, and $X\in C$ with $H\subseteq\chi(X)$:
\[
\mathcal{C},X\models com''(\alpha_H;\beta_H) \;\iff\;
\mathcal{C},X\models com(\alpha_H) \text{ and } \mathcal{C}\otimes\mathcal{U},(X_H,\alpha_H)\models com'(\beta_H).
\]
\end{lemma}
\begin{proof}
Fix $a\in H$. Unfolding $com''(\alpha_a;\beta_a)=\langle\alpha_a\rangle com'(\beta_a)$ via \autoref{lem:diamond} at $H=\{a\}$:
\[
\mathcal C,X\models com''(\alpha_a;\beta_a) \iff \mathcal C,X\models E_{\{a\}} \text{ and } \mathcal C,X_{\{a\}}\models com(\alpha_a) \text{ and } \mathcal C\otimes\mathcal U,(X_{\{a\}},\alpha_a)\models com'(\beta_a).
\]
Since $a\in H\subseteq\chi(X)$, $E_{\{a\}}$ holds trivially. Since $com(\alpha_a)\in\mathcal L^{DDSL}_{\{a\}}$ and $v_a$ is the shared $a$-vertex of $X_{\{a\}}$, $X_H$, and $X$ (all containing $v_a$), \autoref{lem:groupstabDDSL} (at $H=\{a\}$) gives $\mathcal C,X_{\{a\}}\models com(\alpha_a)\iff\mathcal C,X\models com(\alpha_a)\iff\mathcal C,X_H\models com(\alpha_a)$.

For the second conjunct: $(X_{\{a\}},\alpha_a)$ and $(X_H,\alpha_H)$ share the vertex-pair $(v_a,\alpha_a)$ in $\mathcal C\otimes\mathcal U$ (as $\{a\}\subseteq H$), so $\{a\}\subseteq\chi^u\big((X_{\{a\}},\alpha_a)\cap(X_H,\alpha_H)\big)$; since $com'(\beta_a)\in\mathcal L^{DDSL}_{\{a\}}$, \autoref{lem:groupstabDDSL} gives $\mathcal C\otimes\mathcal U,(X_{\{a\}},\alpha_a)\models com'(\beta_a)\iff\mathcal C\otimes\mathcal U,(X_H,\alpha_H)\models com'(\beta_a)$.
Taking the conjunction over $a\in H$ of both established equivalences recovers exactly the claim.
\end{proof}

\begin{remark}
\label{rem:pi-eval}
Since $com(\alpha_a)\in\mathcal L^{DDSL}_{\{a\}}$ for each $a\in H$ (\autoref{lem:strat-mono-ddsl} gives $com(\alpha_H)\in\mathcal L^{DDSL}_H$), and $\chi(X\cap X_H)=H=\chi(X_H)$ whenever $X_H$ is defined (i.e.\ $H\subseteq\chi(X)$), \autoref{lem:groupstabDDSL} gives
\[
\mathcal C,X\models com(\alpha_H) \iff \mathcal C,X_H\models com(\alpha_H).
\]
So $\pi(\alpha_H)=E_H\wedge com(\alpha_H)$ may be evaluated indifferently at $X$ or at $X_H$, justifying the shorthand "$\mathcal C,X\models\pi(\alpha_H)$" used throughout the soundness proofs below.
\end{remark}

We can now move to the soundness and completeness proof of the dynamic fragment of the logic. The main goal is to reduce  
%To prove the soundness and completeness of the dynamic fragment of the logic, we reduce 
every instance of the new dynamic operator to a formula that does not contain it, using the following reduction axioms:

\begin{definition}[Reduction axioms]
\label{def:RAX}
For $\varnothing\ne H,K\subseteq\agents$ and $\varphi\in\mathcal L^{DDSL}$:
\begin{itemize}
   \item Axioms of the static logic \textbf{DSL} (\autoref{def:AX})
   \item \textbf{R1}: $[\alpha_H]p_a \Longleftrightarrow \pi(\alpha_H)\to p_a$, for $a\in H$
   \item \textbf{R2}: $[\alpha_H]\neg\varphi \Longleftrightarrow \pi(\alpha_H)\to\neg[\alpha_H]\varphi$
   \item \textbf{R3}: $[\alpha_H](\varphi\wedge\theta) \Longleftrightarrow [\alpha_H]\varphi\wedge[\alpha_H]\theta$
   \item \textbf{R4a}: $[\alpha_H]\mathbf D_K\varphi \Longleftrightarrow \pi(\alpha_H)\to\mathbf D_K[\alpha_H]\varphi$, for $K\subseteq H$, $\varphi\in\mathcal L^{DDSL}_K$
   \item \textbf{R4b}: $[\alpha_H]\mathbf D_K\varphi \Longleftrightarrow \neg\pi(\alpha_H)$, for $K\not\subseteq H$
   \item \textbf{R5}: $[\alpha_H][\beta_H]\varphi \Longleftrightarrow [\alpha_H;\beta_H]\varphi$
\end{itemize}
\end{definition}

These axioms resemble the ones of standard action model logic \cite{ditmarsch2007dynamic}. 
The interpretation of \textbf{R1-R3} is straightforward. 
Axioms \textbf{R4} represents the interaction between commitment and commitment updates: 
\textbf{R4a} says that, whenever the group ($K$) having a joint commitment is among those enacting the update ($K\subseteq H$), $K$'s post-update commitment to $\varphi$ can equivalently be checked \emph{before} the update: $K$ is already committed now to "$\varphi$ will hold once $\alpha_H$ is enacted", provided the update is executable.
%states that if agents in group $K \subseteq H$, have mutual commitment towards $\varphi$ after action $\alpha_H$, then the presence of agents in $K$ implies that they have a mutual commitment towards joint action $\alpha_H$.
\textbf{R4b} covers the case where the acting group $H$ leaves out some agent needed for the commitment being asked about ($K\not\subseteq H$): since the update can only ever touch $H$'s vertices (\autoref{lem:restriction}), a $K$-commitment with $K\not\subseteq H$ can never be witnessed by the result. So $[\alpha_H]\mathbf D_K\varphi$ is true exactly when $\alpha_H$ fails to execute at all, and false whenever it does, regardless of $\varphi$.
This split mirrors the observation from \autoref{fig:bigtri}: joint commitment is not reducible to the commitments of any subgroup, so an action's effects on a group-level commitment must be handled differently depending on whether the acting group covers it or not.
%covers the case where the acting group $H$ does not include everyone needed for the commitment being asked about ($K \not\subseteq H$) and it reads "after $H$'s action, group $K$ ends up jointly committed to $\varphi$ is equivalent to $H$'s action doesn't actually execute", making the formula vacuously true.
Axiom \textbf{R5} reduces the composition of updates models to the application of sequential updates models.

\begin{theorem}[Soundness]
The reduction axioms of \autoref{def:RAX} are sound.
\end{theorem}

\begin{proof}
\textbf{R1.} By \autoref{def:DDSLsem}: $\mathcal C,X\models[\alpha_H]p_a \iff \mathcal C,X\models\pi(\alpha_H)$ implies $\mathcal C\otimes\mathcal U,(X_H,\alpha_H)\models p_a$. Assume the antecedent: then $a\in H\subseteq\chi(X_H)$ (\autoref{lem:restriction}), and by $l^u(v,\alpha_a)=l(v)$ (\autoref{def:SPU}), $p_a\in l^u(X_H,\alpha_H)\iff p_a\in l(v_a)\iff p_a\in l(X)$, where $v_a$ is the unique $a$-vertex of $X_H$, hence of $X$. So the consequent matches $\mathcal C,X\models p_a$, giving $\mathcal C,X\models[\alpha_H]p_a\iff\mathcal C,X\models\pi(\alpha_H)\to p_a$.

\textbf{R2.} $\mathcal C,X\models[\alpha_H]\neg\varphi \iff \mathcal C,X\not\models\pi(\alpha_H)$ or $(\mathcal C,X\models\pi(\alpha_H)$ and $\mathcal C\otimes\mathcal U,(X_H,\alpha_H)\models\neg\varphi)$ (\autoref{def:DDSLsem}, prop.\ reas.) $\iff \mathcal C,X\not\models\pi(\alpha_H)$ or $\mathcal C,X\not\models[\alpha_H]\varphi \iff \mathcal C,X\models\pi(\alpha_H)\to\neg[\alpha_H]\varphi$.

\textbf{R3.} $\mathcal C,X\models[\alpha_H](\varphi\wedge\theta) \iff \mathcal C,X\models\pi(\alpha_H)$ implies $\mathcal C\otimes\mathcal U,(X_H,\alpha_H)\models\varphi\wedge\theta$ $\iff$ (prop.\ reas., $P\to(Q\wedge R)\equiv(P\to Q)\wedge(P\to R)$) $\mathcal C,X\models[\alpha_H]\varphi\wedge[\alpha_H]\theta$.

\textbf{R4a.} Assume $\mathcal C,X\models\pi(\alpha_H)$ (otherwise both sides are vacuous). Since $K\subseteq H=\chi^u(X_H,\alpha_H)$, the gate for $\mathbf D_K$ holds at $(X_H,\alpha_H)$ in $\mathcal C\otimes\mathcal U$; every witness $(Z,\beta)\in C^u$ with $K\subseteq\chi^u((X_H,\alpha_H)\cap(Z,\beta))$ shares, for each $a\in K$, the vertex-pair $(v_a,\alpha_a)$, so by \autoref{lem:groupstabDDSL} (applied in $\mathcal C\otimes\mathcal U$, at level $K$) all such witnesses agree with $(X_K,\alpha_K)$ on $\varphi\in\mathcal L^{DDSL}_K$:
\[
\mathcal C\otimes\mathcal U,(X_H,\alpha_H)\models\mathbf D_K\varphi \iff \mathcal C\otimes\mathcal U,(X_K,\alpha_K)\models\varphi.\tag{$\ast$}
\]
Separately, $\mathcal C,X\models\mathbf D_K[\alpha_H]\varphi$ (gate holds since $K\subseteq H\subseteq\chi(X)$) reduces, by the same stability argument applied to every witness $Z$ with $K\subseteq\chi(X\cap Z)$ (taking $Z=X$), to $\mathcal C\otimes\mathcal U,(X_K,\alpha_K)\models\varphi$ as well. Combined with $(\ast)$ and \autoref{def:DDSLsem}: $\mathcal C,X\models[\alpha_H]\mathbf D_K\varphi \iff \mathcal C,X\models\pi(\alpha_H)\to\mathbf D_K[\alpha_H]\varphi$.

\textbf{R4b.} If $\mathcal C,X\not\models\pi(\alpha_H)$: LHS vacuously true, RHS $\neg\pi(\alpha_H)$ true. If $\mathcal C,X\models\pi(\alpha_H)$: then $(X_H,\alpha_H)\in C^u$ with color set $H$; since $K\not\subseteq H$, the gate for $\mathbf D_K$ fails there, so $\mathcal C\otimes\mathcal U,(X_H,\alpha_H)\not\models\mathbf D_K\varphi$, making $[\alpha_H]\mathbf D_K\varphi$ false at $X$, matching $\neg\pi(\alpha_H)$, false under this hypothesis.

\item[\textbf{R5.}] $\mathcal{C},X \models [\alpha_H][\beta_H] \varphi \iff \mathcal{C},X \models\pi(\alpha_H) \text{ implies } \mathcal{C}\otimes\mathcal{U},(X_H,\alpha_H) \models [\beta_H] \varphi$ (\autoref{def:DDSLsem})

$\iff \mathcal{C},X\models\pi(\alpha_H)$ implies $\big(\mathcal C\otimes\mathcal U,(X_H,\alpha_H)\models\pi'(\beta_H)$ implies $\mathcal C\otimes\mathcal U\otimes\mathcal U',(X_H,\alpha_H,\beta_H)\models\varphi\big)$ (\autoref{def:DDSLsem} applied within $\mathcal C\otimes\mathcal U$; note $\chi^u(X_H,\alpha_H)=H$)

$\iff \big(\mathcal{C},X\models\pi(\alpha_H) \text{ and } \mathcal C\otimes\mathcal U,(X_H,\alpha_H)\models\pi'(\beta_H)\big)$ implies $\mathcal C\otimes\mathcal U\otimes\mathcal U',(X_H,\alpha_H,\beta_H)\models\varphi$ (prop.\ reas.)

$\iff \mathcal{C},X\models\pi''(\alpha_H;\beta_H)$ implies $\mathcal C\otimes\mathcal U\otimes\mathcal U',(X_H,\alpha_H,\beta_H)\models\varphi$ (\autoref{lem:comEquiv})

$\iff \mathcal{C},X\models[\alpha_H;\beta_H]\varphi$ (\autoref{def:DDSLsem}, identifying $\mathcal C\otimes\mathcal U\otimes\mathcal U'$ with $\mathcal C\otimes(\mathcal U;\mathcal U')$ at the corresponding point, associativity of the product update, as standard for DEL \cite{ditmarsch2007dynamic})
\end{proof}

We now move to the completeness proof, for which we introduce the following translation:

\begin{definition}[Translation]
\label{def:trans}
\begin{itemize}
    \item $t(\top)=\top$, \quad $t(p_a)=p_a$, \quad $t(\neg\varphi)=\neg t(\varphi)$, \quad $t(\varphi\wedge\psi)=t(\varphi)\wedge t(\psi)$, \quad $t(\mathbf D_G\varphi)=\mathbf D_G t(\varphi)$
    \item $t([\alpha_H]p_a) = t(\pi(\alpha_H)\to p_a)$
    \item $t([\alpha_H]\neg\varphi) = t(\pi(\alpha_H)\to\neg[\alpha_H]\varphi)$
    \item $t([\alpha_H](\varphi\wedge\theta)) = t([\alpha_H]\varphi\wedge[\alpha_H]\theta)$
    \item $t([\alpha_H]\mathbf D_K\varphi) = \begin{cases} t(\pi(\alpha_H)\to\mathbf D_K[\alpha_H]\varphi) & K\subseteq H \\ t(\neg\pi(\alpha_H)) & K\not\subseteq H\end{cases}$
    \item $t([\alpha_H][\beta_H]\varphi) = t([\alpha_H;\beta_H]\varphi)$
\end{itemize}
\end{definition}

The core idea of a reduction proof is to show that the complexity of the argument on the LHS is greater than the complexity of the RHS. 
In order to correctly parametrize the complexity for the composition axiom \textbf{R5}, we need to introduce some auxiliary definitions. The remainder of this subsection is dedicated entirely to showing that repeatedly applying \textbf{R1}--\textbf{R5} always eventually eliminates every action modality. A reader interested in what the axioms mean, rather than why the translation is well-founded, may skip ahead to \autoref{translation}.

\begin{definition}[Action size]
\label{def:actsize}
For a vertex $\gamma\in V(E)$ of any update model $\mathcal U$ (\autoref{def:commUmod}), define $\|\gamma\|\in\mathbb N_{\ge1}$ by structural recursion on how $\gamma$ was built:
\begin{itemize}
    \item If $\mathcal U$ is a \emph{primitive} update model (i.e.\ not itself constructed via \autoref{def:compo}, so $com(\gamma)\in\mathcal L^{DSL}_{\{a\}}$ for $\gamma$ of color $a$): $\|\gamma\|:=1$.
    \item If $\gamma=\alpha_a;\beta_a\in V(E'')$ (\autoref{def:compo}): $\|\alpha_a;\beta_a\|:=\|\alpha_a\|+\|\beta_a\|$.
\end{itemize}
For a face $\alpha_H=\{\alpha_a\}_{a\in H}$, extend additively over the group: $\|\alpha_H\|:=\sum_{a\in H}\|\alpha_a\|$.
\end{definition}

In order to account for the reduction of consecutive commitment updates \textbf{R5} we introduce two complexity measures:

\begin{definition}[Primary complexity $\delta$]
\label{def:delta}
\begin{align*}
\delta(\top)=\delta(p_a) &= 0 \\
\delta(\neg\varphi) &= \delta(\varphi) \\
\delta(\varphi\wedge\psi) = \delta(\varphi\to\psi) &= \delta(\varphi)+\delta(\psi) \\
\delta(\mathbf D_G\varphi) &= \delta(\varphi) \\
\delta([\alpha_H]\varphi) &= \|\alpha_H\| + \delta(\varphi)
\end{align*}
\end{definition}

\begin{definition}[Secondary complexity $n$]
\label{def:ncx}
\begin{align*}
n(\top)=n(p_a) &= 1 \\
n(\neg\varphi) = n(\mathbf D_G\varphi) = n([\alpha_H]\varphi) &= 1+n(\varphi) \\
n(\varphi\wedge\psi) = n(\varphi\to\psi) &= 1+\max(n(\varphi),n(\psi))
\end{align*}
\end{definition}

Note that $n$ is completely uniform across $\neg$, $\mathbf{D}_G$, and $[\alpha_H]$.
While $\delta$  owns the "how much action-unfolding is left" dimension, $n$ only ever needs to break ties between formulas built from the same action, where the action itself never varies.
The overall order is $(\delta,n)$
 compared lexicographically: $(\delta_1,n_1) < (\delta_2,n_2)$ iff $\delta_1<\delta_2$, or $\delta_1=\delta_2$ and $n_1<n_2$. This is well-founded since it's a lexicographic product of two copies of $(\mathbb N,<)$.

\begin{lemma}[Commitment formulas are strictly simpler than their action]
\label{lem:comdepth}
For every vertex $\gamma$ of any update model (primitive or composed), $\delta(com(\gamma)) < \|\gamma\|$, where $\delta$ is defined on formulas in \autoref{def:delta}.
\end{lemma}

\begin{proof}
By strong induction on $\|\gamma\|$.

\emph{Base case, $\|\gamma\|=1$}: then $\gamma$ is primitive (composed vertices always have $\|\cdot\|\ge2$, being a sum of two terms each $\ge1$), so $com(\gamma)\in\mathcal L^{DSL}_{\{a\}}$, hence $\delta(com(\gamma))=0<1$.

\emph{Inductive case}: $\gamma=\mu;\nu$ with $\|\gamma\|=\|\mu\|+\|\nu\|$. By \autoref{def:compo}, $com(\gamma)=\langle\mu\rangle com_\nu(\nu)=\neg[\mu]\neg com_\nu(\nu)$, so (using \autoref{def:delta} below, which gives $\delta(\neg\chi)=\delta(\chi)$ and $\delta([\mu]\chi)=\|\mu\|+\delta(\chi)$):
\[
\delta(com(\gamma)) = \|\mu\| + \delta(com_\nu(\nu)).
\]
Since $\|\nu\|<\|\gamma\|$, the induction hypothesis applies to $\nu$: $\delta(com_\nu(\nu))<\|\nu\|$. Hence $\delta(com(\gamma)) = \|\mu\|+\delta(com_\nu(\nu)) < \|\mu\|+\|\nu\| = \|\gamma\|$.
\end{proof}

The intuitive meaning behind this lemma is that no matter how deep the composition, the commitment formula attached to an action is always built from strictly fewer primitive actions than the action itself has.

\begin{theorem}[Well-foundedness of the reduction axioms]
\label{thm:wf}
For every instance of \textbf{R1}–\textbf{R5}, every formula \(t\) recurses into on the right-hand side has strictly smaller $(\delta,n)$ than the left-hand side.
\end{theorem}

\begin{proof}

The proof is structured checking the complexity measure $\delta$ first, and then checking $n$ in case of ties.
We use the following, which holds by\autoref{lem:comdepth}:

\[
\delta(\pi(\alpha_H)) = \delta(E_H)+\delta(com(\alpha_H)) = 0+\delta(com(\alpha_H)) < \|\alpha_H\| \le \|\alpha_H\|+\delta(\varphi) = \delta([\alpha_H]\varphi) \tag{P}
\]

\textbf{R1}: $[\alpha_H]p_a\to \pi(\alpha_H)\to p_a$. Only nontrivial recursive target is $\pi(\alpha_H)$: strict by (P), since $\delta([\alpha_H]p_a)=\|\alpha_H\|+\delta(p_a)=\|\alpha_H\|$.

\textbf{R2}: $[\alpha_H]\neg\varphi\to\pi(\alpha_H)\to\neg[\alpha_H]\varphi$. Recursive targets: $\pi(\alpha_H)$ (strict by (P)) and $[\alpha_H]\varphi$. For the latter: $\delta([\alpha_H]\varphi)=\|\alpha_H\|+\delta(\varphi)=\|\alpha_H\|+\delta(\neg\varphi)=\delta([\alpha_H]\neg\varphi)$. 
Tie-break: $n([\alpha_H]\varphi)=1+n(\varphi)$ versus $n([\alpha_H]\neg\varphi)=1+n(\neg\varphi)=1+(1+n(\varphi))=2+n(\varphi)$, resulting in: $1+n(\varphi)<2+n(\varphi)$.

\textbf{R3}: $[\alpha_H](\varphi\wedge\theta)\to[\alpha_H]\varphi\wedge[\alpha_H]\theta$. Recursive targets $[\alpha_H]\varphi$, $[\alpha_H]\theta$ (since they ara analogous, we check only one). $\delta([\alpha_H]\varphi)=\|\alpha_H\|+\delta(\varphi)$ versus $\delta([\alpha_H](\varphi\wedge\theta))=\|\alpha_H\|+\delta(\varphi)+\delta(\theta)$: tied iff $\delta(\theta)=0$, else strict decrease already. If tied ($\delta(\theta)=0$): $n([\alpha_H]\varphi)=1+n(\varphi)$ versus $n([\alpha_H](\varphi\wedge\theta))=1+n(\varphi\wedge\theta)=1+1+\max(n(\varphi),n(\theta))=2+\max(n(\varphi),n(\theta))\ge2+n(\varphi)$.

\textbf{R4a}: $[\alpha_H]\mathbf D_K\varphi\to\pi(\alpha_H)\to\mathbf D_K[\alpha_H]\varphi$ ($K\subseteq H$). Recursive targets: $\pi(\alpha_H)$ (strict by (P)) and, via $t(\mathbf D_K\psi)=\mathbf D_K t(\psi)$, ultimately $[\alpha_H]\varphi$. $\delta([\alpha_H]\varphi)=\|\alpha_H\|+\delta(\varphi)$; $\delta([\alpha_H]\mathbf D_K\varphi)=\|\alpha_H\|+\delta(\mathbf D_K\varphi)=\|\alpha_H\|+\delta(\varphi)$. Tie-break: $n([\alpha_H]\varphi)=1+n(\varphi)$ versus $n([\alpha_H]\mathbf D_K\varphi)=1+n(\mathbf D_K\varphi)=1+(1+n(\varphi))=2+n(\varphi)$.

\textbf{R4b}: $[\alpha_H]\mathbf D_K\varphi\to\neg\pi(\alpha_H)$ ($K\not\subseteq H$). Only recursive target is $\pi(\alpha_H)$: strict by (P), since $\delta([\alpha_H]\mathbf D_K\varphi)=\|\alpha_H\|+\delta(\varphi)\ge\|\alpha_H\|>\delta(\pi(\alpha_H))$.

\textbf{R5}: $[\alpha_H][\beta_H]\varphi\to[\alpha_H;\beta_H]\varphi$. First notice that
\[
\delta([\alpha_H][\beta_H]\varphi) = \|\alpha_H\|+\delta([\beta_H]\varphi) = \|\alpha_H\|+\|\beta_H\|+\delta(\varphi).
\]
\[
\delta([\alpha_H;\beta_H]\varphi) = \|\alpha_H;\beta_H\|+\delta(\varphi) = (\|\alpha_H\|+\|\beta_H\|)+\delta(\varphi)
\]
using the additivity of $\|\cdot\|$ under composition (\autoref{def:actsize}). 
Tie-break via $n$:
\[
n([\alpha_H][\beta_H]\varphi) = 1+n([\beta_H]\varphi) = 1+(1+n(\varphi)) = 2+n(\varphi),
\]
\[
n([\alpha_H;\beta_H]\varphi) = 1+n(\varphi).
\]
\end{proof}

\begin{theorem}
\label{translation}
For every formula $\varphi\in\mathcal L^{DDSL}$, $\vdash\varphi\leftrightarrow t(\varphi)$.
\end{theorem}
\begin{proof}
By strong induction on $(\delta(\varphi),n(\varphi))$, ordered lexicographically (\autoref{def:delta}, \autoref{def:ncx}). The ordinary structural cases ($\neg$, $\wedge$, $\mathbf D_G$) tie in $\delta$ and strictly decrease in $n$ by exactly 1, exactly as for standard modal completeness proofs. The five dynamic cases are handled by the corresponding axiom of \autoref{def:RAX} (\textbf{R1}–\textbf{R5}), each shown in \autoref{thm:wf} to recurse only into strictly $(\delta,n)$-smaller formulas.
\end{proof}

\begin{lemma}[Translation preserves stratification]
\label{lem:t-preserves-strat}
For every $\varnothing\ne G\subseteq\agents$ and $\varphi\in\mathcal L^{DDSL}_G$: $t(\varphi)\in\mathcal L^{DSL}_G$.
\end{lemma}
\begin{proof}
By induction on $(\delta(\varphi),n(\varphi))$ (\autoref{def:delta}, \autoref{def:ncx}), ordered lexicographically.

$\varphi=\top$: $t(\top)=\top\in\mathcal L^{DSL}_G$.

$\varphi=p_a$ ($a\in G$): $t(p_a)=p_a\in\mathcal L^{DSL}_G$.

$\varphi=\neg\psi$, $\varphi=\psi_1\wedge\psi_2$: immediate from the IH, since $t$ commutes with $\neg,\wedge$ and both cases preserve $G$.

$\varphi=\mathbf D_H\psi$ ($H\subseteq G$, $\psi\in\mathcal L^{DDSL}_H$): $t(\mathbf D_H\psi)=\mathbf D_H t(\psi)$; by IH (applied at level $H$, lower complexity), $t(\psi)\in\mathcal L^{DSL}_H$; since $H\subseteq G$, $\mathbf D_H t(\psi)\in\mathcal L^{DSL}_G$ by the defining clause of \autoref{def:langL}.

$\varphi=[\alpha_H]\theta$ ($H\subseteq G$, $\theta\in\mathcal L^{DDSL}_H$), any of R1--R5: in every case, $t([\alpha_H]\theta)$ is defined as $t(\rho)$ for some $\rho$ built from $com(\alpha_H)$, $\theta$ (or its immediate subformulas), $\mathbf D_H$, and Boolean connectives all at level $H$. Since $com(\alpha_a)\in\mathcal L^{DDSL}_{\{a\}}\subseteq\mathcal L^{DDSL}_H$ for each $a\in H$ (\autoref{lem:strat-mono-ddsl}), $com(\alpha_H)\in\mathcal L^{DDSL}_H$; combined with $\theta\in\mathcal L^{DDSL}_H$, $\rho\in\mathcal L^{DDSL}_H$. By IH (lower complexity), $t(\rho)\in\mathcal L^{DSL}_H\subseteq\mathcal L^{DSL}_G$.
\end{proof}

\begin{theorem}[Completeness] For every formula $\varphi \in \mathcal{L}^{DDSL}$, $\models \varphi \text{ implies } \vdash \varphi$

\end{theorem}

\begin{proof}
Suppose $\textbf{DDSL} \models \varphi$. By the soundness of the proof system and by \autoref{translation}, $\textbf{DSL} \models t(\varphi)$. By \autoref{lem:t-preserves-strat}, $t(\varphi)\in\mathcal L^{DSL}_\agents=\mathcal L^{DSL}$ is well-formed, so \autoref{theo:compl} applies, giving $\textbf{DDSL}\vdash t(\varphi)$. Since $\textbf{DDSL}\vdash \varphi \leftrightarrow t(\varphi)$, $\textbf{DDSL}\vdash \varphi$.
\end{proof}

\section{Conclusion and Further Work}
\label{sec:concl}

This paper introduces the novel deontic simplicial logic (\textbf{DSL}), the first deontic logical framework based on (impure) simplicial complexes. Its geometric interpretation allows us to formalize joint commitments among agents in a natural way.
We further extend \textbf{DSL} to the Dynamic Deontic Simplicial Logic (\textbf{DDSL}), resulting, to our knowledge, the first dynamic \emph{deontic} logic based on simplicial complexes, representing model transforming actions as agents' choices among mutually exclusive (joint) commitments. Both logics are shown to be sound and complete, and their usefulness is illustrated through several examples.

As future work, we plan to extend the framework to accommodate non-symmetric commitments, in which an agent $a$ is committed toward agent $b$ without a corresponding commitment from $b$ to $a$.
This can be achieved by using directed simplicial complexes, where arrow represents commitments from only one of the parties.
We also intend to incorporate graded commitments, attributing values to vertices, edges and faces, thereby enabling a more fine-grained form of normative reasoning. In addition, we aim to explore the representation of other deontic notions, such as permissions and prohibitions, within the simplicial setting.
Finally, the satisfiability/decidability of these logics is an open problem.
%as well as decidability results of the existing simplicial logics.

\section*{Acknowledgments}

We are thankful to Hans van Ditmarsch for the multiple illuminating discussions.

\section*{Funding}
This work was supported by the Austrian Science Fund (FWF) project A Logical Framework for Graded Deontic Reasoning [10.55776/PAT2141924] (\doi{10.55776/PAT2141924}).

%\nocite{*}
\bibliographystyle{plainurl}
\bibliography{references}

\end{document}